\documentclass[12pt,reqno]{amsart}

\usepackage{amsmath,amssymb,amsthm}
\usepackage{amsaddr}
\usepackage{bm}
\usepackage[numbers]{natbib}
\usepackage{here,comment}
\usepackage{algorithm,algorithmic}
\usepackage{xcolor}
\usepackage{fullpage}
\usepackage{booktabs}
\usepackage{graphicx}
\usepackage{subcaption}

\usepackage{mathrsfs}
\usepackage{ascmac}
\usepackage{natbib}
\usepackage{setspace}
\usepackage{url}
\usepackage{booktabs}
\usepackage{stmaryrd}
\usepackage{mathtools}
\mathtoolsset{showonlyrefs}
\usepackage{algorithmic}
\usepackage{algorithm}

\usepackage{bm}
\newtheorem{thm}{Theorem}

\newtheorem{prp}{Proposition}

\newtheorem{rmk}{Remark}

\DeclareMathOperator*{\argmin}{arg\,min}

\renewcommand{\hat}{\widehat}
\renewcommand{\tilde}{\widetilde}

\usepackage{color}

\title{Distribution-on-Distribution Regression with Wasserstein Metric: Multivariate Gaussian Case}
\author{Ryo Okano$^1$ \and Masaaki Imaizumi$^{1,2}$}

\address{$^1$The University of Tokyo\footnote{\textit{Address}:  Hongo 7-3-1, Bunkyo City, Tokyo, JAPAN. 113-8654. \\\textit{Mail}: \url{okano-ryo1134@g.ecc.u-tokyo.ac.jp}, \url{imaizumi@g.ecc.u-tokyo.ac.jp}}, $^2$RIKEN Center for Advanced Intelligence Project}

\date{\today}

\allowdisplaybreaks

\begin{document}

\maketitle

\begin{abstract}
Distribution data refers to a data set where each sample is represented as a probability distribution, a subject area receiving burgeoning interest in the field of statistics. Although several studies have developed distribution-to-distribution regression models for univariate variables, the multivariate scenario remains under-explored due to technical complexities. In this study, we introduce models for regression from one Gaussian distribution to another, utilizing the Wasserstein metric. These models are constructed using the geometry of the Wasserstein space, which enables the transformation of Gaussian distributions into components of a linear matrix space. Owing to their linear regression frameworks, our models are intuitively understandable, and their implementation is simplified because of the optimal transport problem's analytical solution between Gaussian distributions. We also explore a generalization of our models to encompass non-Gaussian scenarios. We establish the convergence rates of in-sample prediction errors for the empirical risk minimizations in our models. In comparative simulation experiments, our models demonstrate superior performance over a simpler alternative method that transforms Gaussian distributions into matrices. We present an application of our methodology using weather data for illustration purposes.
\end{abstract}


\section{Introduction}
The analysis of distribution data has gained significant attention in the field of statistics. Distribution data refers to data in which each sample is given in the form of a probability distribution or an empirical distribution generated from it. Examples include age-at-death distributions across different countries, house price distributions of different years, and distributions of voxel-voxel correlations of functional magnetic imaging signals. A distinctive feature of distribution data is that they take values in general metric spaces that lack a vector space structure. Existing complex data analysis methods, such as function or manifold data analysis methods, are inadequate for effectively handling distribution data due to their infinite dimensionality and non-linearity, posing significant challenges in processing. Developing methods and theories for analyzingdistribution data is an important and challenging problem for contemporary statistical practice.  Refer to \cite{petersen2022modeling} for a review of this topic.

A common approach to handling distribution data involves the application of the Wasserstein metric to a set of distributions.
The resulting metric space is known as the Wasserstein space (\cite{panaretos2020invitation}), where distribution data are considered as its elements.
There are several advantages to using the Wasserstein metric: it gives more intuitive interpretations of mean and geodesics compared to other metrics, and it reduces errors by rigorously treating constraints as distribution functions.
Based on this approach, numerous methods have been proposed for the anlaysis of distribution data (\cite{bigot2017geodesic,petersen2019wasserstein,petersen2019frechet,fan2021conditional,chen2021wasserstein,ghodratidistribution,zhang2022wasserstein}). 

This paper focuses on a problem of distribution-on-distribution regression, that is, the regression of one probability distribution onto another. 
In the distribution-on-distribution regression problem, the task involves defining a regression map between non-linear spaces, which makes this problem technically challenging.
The problem is used for comparing the temporal evolution of age-at-death distributions among different countries (\cite{chen2021wasserstein}, \cite{ghodratidistribution}) and predicting house price distributions in the United States(\cite{chen2021wasserstein}).
For univariate distributions, several studies have investigated distribution-on-distribution regression models using Wasserstein metric.
\cite{chen2021wasserstein} proposed a model utilizing geometric properties of the Wasserstein space, \cite{zhang2022wasserstein} presented an autoregressive model for distributional time series data, and 
\cite{ghodratidistribution}  introduced a model incorporating the optimal transport map associated with the Wasserstein space.
However, few studies proposed distribution-on-distribution regression models for the multivariate case with the Wasserstein metric.
For more detail, please refer to Section \ref{sec:comarison} for a comprehensive overview.

In this paper, we propose models for regressing one Gaussian distribution onto another. 
To define our models, we consider the space of Gaussian distributions equipped with the Wasserstein metric and use its tangent bundle structure to transform Gaussian distributions into matrices. 
Then, we boil down the Gaussian distribution-on-distribution regression to the matrix-on-matrix linear regression, using the transformation to the tangent bundle.
Based on the transformation, we proposed two models: a basic model for the case where predictor and response Gaussian distributions are low-dimensional, and a low-rank model incorporating a low-rank structure in the parameter tensor to address high-dimensional Gaussian distributions.
Additionally, we explore the extension of our proposed models to encompass non-Gaussian scenarios.

Our strategy and the model give several advantages:
(i) the strategy enables the explicit construction of regression maps using the closed-form expression for the optimal transport problem between Gaussian distributions,
(ii) it boils down the distribution-on-distribution regression problem to an easy-to-handle linear model while maintaining the constraint of distributions, and (iii) we can solve the linear model without computational difficulties.
The effectiveness of our approach is also demonstrated through simulations. In particular, in comparison to the matrix-on-matrix regression model without the Wassetstein metric, our approach achieves better accuracy, taking advantage of the use of the Wasserstein metric.

The remaining sections of the paper are organized as follows.
In Section \ref{sec:background}, we provide some background on the optimal transport and Wasserstein space.
In Section \ref{sec:model}, we introduce Gaussian distribution-on-distribution regression models and discuss their potential generalizations to accommodate non-Gaussian cases.
We show empirical risk minimization algorithms in our models in Section \ref{sec:estimator}, and analyze their in-sample prediction errors in Section \ref{sec:prediction_error}. 
We investigate the finite-sample performance of the proposed methods through simulation studies in Section \ref{sec:simulation}, and illustrate the application of the proposed method using weather data in Section \ref{sec:application}.
Section \ref{sec:conclusion} concludes.
Proofs of theorems and additional theoretical results are provided in Appendix.

\subsection{Related Studies}
There are several approaches to deal with distribution data apart from the Wasserstein metric approach.
\cite{petersen2016functional} introduced the log quantile density transformation, enabling the utilization of functional data methods for distribution data.
The Bayes space approach has also been proposed as a viable solution for handling distribution data (\cite{egozcue2006hilbert,van2014bayes,talska2020weighting}).

Within the framework of the Wasserstein metric approach, significant developments have been made in methods and theories for analyzing distribution data.
\cite{zemel2019frechet} considered the estimation for the Fr\'{e}chet mean, a notion of mean in the Wasserstein space, from distribution samples. 
\cite{bigot2018upper} established the minimax rates of convergence for these estimators.
\cite{petersen2019wasserstein} proposed the Wasserstein covariance measure for dependent density data.
\cite{bigot2017geodesic} developed the method of geodesic principal component analysis on the Wasserstein space.

Various regression models utilizing the Wasserstein metric have been proposed for distribution data. 
\cite{petersen2019frechet} developed regression models for coupled vector predictors and univariate random distributions as responses.
\cite{fan2021conditional} developed regression models for multivariate response distributions.
\cite{chen2021wasserstein} and \cite{ghodratidistribution} 
proposed regression models for scenarios where both regressors and responses are random distributions, and \cite{ghodrati2023transportation} studies its extension to the multivariate case. 
\cite{zhang2022wasserstein}
developed autoregressive models for density time series data.

\subsection{Notation}
\label{subsec:notation}
For $d \ge 1$, we denote the identity matrix of size $d \times d$ as $I_d$. 
$\text{Sym}(d)$ is a set of all symmetric matrices of size $d \times d$.
For a positive semidefinite matrix $A$, we denote its positive square root as $A^{1/2}$.
$\text{id}(\cdot)$ is the  identity map. For a Borel measurable function $f: \mathbb{R}^d \to \mathbb{R}^d$ and  Borel probability measure $\mu$ on $\mathbb{R}^d$, $f\#\mu$ is the push-forward measure defined by $f\#\mu(\Omega) = \mu(f^{-1}(\Omega))$ for any Borel set $\Omega$ in $\mathbb{R}^d$.
$\|\cdot \|$ denotes the  Euclidean norm.
${\mathcal{L}_{\mu}^2(\mathbb{R}^d)}$ is the sef of functions $f:\mathbb{R}^d \to \mathbb{R}^d$ such that $\int\|f(x)\|^2 d\mu(x) < \infty$, and is a Hilbert space with an inner product $\langle \cdot, \cdot \rangle_{\mu}$ defined as $\langle f, g \rangle_{\mu} = \int_{\mathbb{R}^d} 
f(x)^\top g(x)d\mu(x)$ for $f,g \in {\mathcal{L}_{\mu}^2(\mathbb{R}^d)}$.
We denote the norm induced by this inner product as $\|\cdot\|_{\mu}$.

For a matrix $A \in \mathbb{R}^{d_1 \times d_2}$, we denote its elements as $A[p, q]$ for $1 \le p \le d_1$ and $1 \le q \le d_2$.
For a tensor $\mathbb{A} \in \mathbb{R}^{d_1 \times d_2 \times d_3 \times d_4}$, we denote its
elements as 
$\mathbb{A}[p, q, r, s]$ 
for $1 \le p \le d_1, 1 \le q \le d_2, 1 \le r \le d_3$ and $1 \le s \le d_4$.
For a tensor $\mathbb{A} \in \mathbb{R}^{d_1 \times d_2 \times d_3 \times d_4}$ and indices $1 \le r \le d_3, 1 \le s \le d_4$, let
$\mathbb{A}[\cdot, \cdot, r, s] \in \mathbb{R}^{d_1 \times d_2}$ denote the $d_1 \times d_2$ matrix whose $(p, q)$-elements are given by $\mathbb{A}[p, q, r, s]$.
Likewise, for indices $1 \le p \le d_1, 1 \le q \le d_2$, $\mathbb{A}[p, q, \cdot, \cdot] \in \mathbb{R}^{d_3 \times d_4}$ denote the  $d_3 \times d_4$ matrix whose $(r, s)$-elements are given by $\mathbb{A}[p, q, r, s]$.
For vectors $a_1 \in \mathbb{R}^{d_1}, a_2 \in \mathbb{R}^{d_2}, a_3 \in \mathbb{R}^{d_3}$ and $a_4 \in \mathbb{R}^{d_4}$, let define the outer product $\mathbb{A} = a_1 \circ a_2 \circ a_3 \circ a_4 \in \mathbb{R}^{d_1 \times d_2 \times d_3 \times d_4}$ by
$\mathbb{A}[p, q, r, s] = a_1[p]a_2[q]a_3[r]a_4[s]$.
For two matrices $A_1, A_2 \in \mathbb{R}^{d_1 \times d_2}$, we define their inner product $\langle A_1, A_2 \rangle \in \mathbb{R}$ as $\langle A_1, A_2 \rangle = \sum_{p=1}^{d_1}\sum_{q=1}^{d_2}A_1[p, q]A_2[p, q]$. Furthermore, for a tensor
$\mathbb{A} \in \mathbb{R}^{d_1 \times d_2 \times d_3 \times d_4}$
and a matrix $A \in \mathbb{R}^{d_1 \times d_2}$, we define their product $\langle A, \mathbb{A} \rangle_2 \in \mathbb{R}^{d_3 \times d_4}$ as 
$\langle A, \mathbb{A} \rangle_2[r, s] = \sum_{p=1}^{d_1}\sum_{q=1}^{d_2}A[p,q]\mathbb{A}[p,q,r,s]$ for $1 \le r \le d_3$ and $1 \le s \le d_4$.

\section{Background} \label{sec:background}
In this section, we provide some background on optimal transport, the Wasserstein space, and its tangent space.
For more background, see e.g., \cite{villani2009optimal}, \cite{ambrosio2005gradient} and \cite{panaretos2020invitation}.
\subsection{Optimal Transport}
Let $\mathcal{W}(\mathbb{R}^d)$ be the set of Borel probability distributions on $\mathbb{R}^d$ with finite second moments. The 2-Wasserstein distance between $\mu_1, \mu_2 \in \mathcal{W}(\mathbb{R}^d)$ is defined by 
\begin{align}
    d_W(\mu_1, \mu_2)
=
\left(\inf_{\pi \in \Pi(\mu_1, \mu_2)} \int_{\mathbb{R}^d \times \mathbb{R}^d} \|x - y\|^2d\pi(x, y)\right)^{1/2}.
\label{eq:wasserstein}
\end{align}
Here, $\Pi(\mu_1, \mu_2)$ is the set of couplings of $\mu_1$ and $\mu_2$, that is, 
the set of joint distributions on $\mathbb{R}^d \times \mathbb{R}^d$ with marginal distributions $\mu_1$ and $\mu_2$.
In our setting, the minimizer $\pi$ in \eqref{eq:wasserstein}
always exists (Theorem 4.1 in \cite{villani2009optimal}), and is called an optimal coupling.
When $\mu_1$ is absolutely continuous with respect to the Lebesgue measure,
there exists a map $T: \mathbb{R}^d \to \mathbb{R}^d$ such that the joint distribution of $(\Bar{W}, T(\Bar{W}))$, where $\Bar{W} \sim \mu_1$,  is an optimal coupling in \eqref{eq:wasserstein}, and
such a map $T$ is uniquely determined $\mu_1$-almost everywhere (Theorem 1.6.2 in \cite{panaretos2020invitation}). The map $T$ is called the optimal transport map between $\mu_1$ and $\mu_2$, and
we denote it as $T_{\mu_1}^{\mu_2}$. 
When $d=1$, the optimal transport map has the following closed-form expression (Section 1.5 in \cite{panaretos2020invitation}):
\begin{equation}
    T_{\mu_1}^{\mu_2}(x)
    =
    F_{\mu_2}^{-1} \circ F_{\mu_1}(x), \quad x \in \mathbb{R}, 
    \label{eq:ot_onedim}
\end{equation}
where $F_{\mu_1}$ is the cumulative distribution function of $\mu_1$, and 
$F_{\mu_2}^{-1}$ is the quantile funciton of $\mu_2$.

\subsection{The Wasserstein Space and its Tangent Space}
\label{sec:tangent}
The Wasserstein distance $d_W$ is a metric on $\mathcal{W}(\mathbb{R}^d)$ (Chapter 6 in \cite{villani2009optimal}), and the metric space $(\mathcal{W}(\mathbb{R}^d), d_W)$ is called the Wasserstein space.
We give a notion of a linear space induced from the Wasserstein space, by applying the 
the basic concepts of Riemannian manifolds, as shown in \cite{ambrosio2005gradient}, \cite{bigot2017geodesic} and \cite{panaretos2020invitation}.

Let arbitrarily fix a reference measure $\mu_\ast \in \mathcal{W}(\mathbb{R}^d)$ which is 
absolutely continuous with 
respect to 
 the Lebesgue measure.
For any $\mu \in \mathcal{W}(\mathbb{R}^d)$, the geodesic 
from $\mu_\ast$ to $\mu$, $\gamma_{\mu_\ast, \mu}: [0, 1] \to \mathcal{W}(\mathbb{R}^d)$, is given by 
\begin{align}
    \gamma_{\mu_\ast, \mu}(t)
=
[t(T_{\mu_\ast}^{\mu} - \text{id}) + \text{id}]\#\mu_\ast, \quad t \in [0, 1].
\label{eq:geodesic}
\end{align}
The tangent space of the Wasserstein space at $\mu_\ast$ is defined by 
\begin{align}
    \mathcal{T}_{\mu_\ast}
=
\overline{
\{t(T_{\mu_\ast}^\mu - \text{id}): \mu \in \mathcal{W}(\mathbb{R}^d), t > 0 \}
},
\label{eq:tanent_space}
\end{align}
where the upper bar denotes the closure in terms of the norm $\|\cdot\|_{\mu_\ast}$ in the space $\mathcal{L}_{\mu_\ast}^2(\mathbb{R}^d)$.
The space $ \mathcal{T}_{\mu_\ast}$ is a subspace of ${\mathcal{L}_{\mu_\ast}^2(\mathbb{R}^d)}$ (Theorem 8.5.1 in \cite{ambrosio2005gradient}).
The exponential map $\text{Exp}_{\mu_\ast}: \mathcal{T}_{\mu_\ast} \to \mathcal{W}(\mathbb{R}^d)$ is then defined by 
\begin{align}
    \text{Exp}_{\mu_\ast}g
    =
    (g + \text{id}) \# \mu_\ast, \quad g \in \mathcal{T}_{\mu_\ast},
    \label{eq:exp_map}
\end{align}
and as its right inverse, the logarithmic map $\text{Log}_{\mu_\ast}: \mathcal{W}(\mathbb{R}^d) \to \mathcal{T}_{\mu_\ast}$ is given by 
\begin{align}
    \text{Log}_{\mu_\ast}\mu = T_{\mu_\ast}^\mu - \text{id}, \quad \mu \in \mathcal{W}(\mathbb{R}^d).
    \label{eq:log_map}
\end{align}
When $d=1$, 
the logarithmic map is isometric in the sense that 
\begin{equation}
    \|\text{Log}_{\mu_\ast}\mu_1 - \text{Log}_{\mu_\ast}\mu_2\|_{\mu_\ast} = 
d_W(\mu_1, \mu_2)
\label{eq:isomet_onedim}
\end{equation}
for all $\mu_1, \mu_2 \in \mathcal{W}(\mathbb{R})$ (Section 2.3.2 in \cite{panaretos2020invitation}).
Remind that $\|\cdot\|_{\mu^*}$ is the norm of ${\mathcal{L}_{\mu_\ast}^2(\mathbb{R}^d)}$ with the reference measure $\mu^\ast$, as defined in Section \ref{subsec:notation}.

\subsection{Specification with Gaussian Case}

We restrict our attention to the Gaussian measures. 
Let $\mathcal{G}(\mathbb{R}^d)$ be the set of Gaussian distributions on $\mathbb{R}^d$, and we call the metric space $(\mathcal{G}(\mathbb{R}^d), d_W)$ as the Gaussian space.

For two Gaussian measures
 $\mu_1 = N(m_1, \Sigma_1)$, $\mu_2 = N(m_2, \Sigma_2) \in \mathcal{G}(\mathbb{R}^d)$ with mean vectors $m_1,m_2 \in \mathbb{R}^d$ and covariance matrices $\Sigma_1,\Sigma_2 \in \mathbb{R}^{d\times d}$, the 
2-Wasserstein distance between them has the following closed-form expression (Section 1.6.3 in \cite{panaretos2020invitation}):
\begin{equation}
    d_W(\mu_1, \mu_2)
    =
    \sqrt{
    \|m_1 - m_2\|^2 + \text{tr}[\Sigma_1 + \Sigma_2 - 2(\Sigma_1^{1/2}\Sigma_2 \Sigma_1^{1/2})^{1/2}]}.
    \label{eq:W_formula}
\end{equation}
When $\Sigma_1$ is non-singular,
the optimal transport map between $\mu_1$ and $\mu_2$ also has the following closed-form expression (Section 1.6.3 in \cite{panaretos2020invitation}):
\begin{equation}
     T_{\mu_1}^{\mu_2}(x)
    = m_2 + S(\Sigma_1,\Sigma_2)(x-m_1), \quad x \in \mathbb{R}^d,
    \label{eq:OT_formula}
\end{equation}
where we define $S(\Sigma_1,\Sigma_2) = \Sigma_1^{-1/2}[\Sigma_1^{1/2} \Sigma_2 \Sigma_1^{1/2}]^{1/2} \Sigma_1^{-1/2}$ for two covariance matrices $\Sigma_1, \Sigma_2$.

We introduce a tangent space of Gaussian spaces.
Fix a Gaussian measure $\mu_\ast = N(m_\ast, \Sigma_\ast) \in \mathcal{G}(\mathbb{R}^d)$ as a reference measure with a non-singular covariance matrix $\Sigma_*$.
Replacing $\mathcal{W}(\mathbb{R}^d)$ with $\mathcal{G}(\mathbb{R}^d)$ in the definition of tangent space \eqref{eq:tanent_space}, we obtain the tangent space by a form of a function space
\begin{equation}
    \mathcal{TG}_{\mu_\ast} 
    =
\overline{
\{t(T_{\mu_\ast}^\mu - \text{id}): \mu \in \mathcal{G}(\mathbb{R}^d), t > 0 \}}.
\end{equation}
Using the form of the optimal transport map
\eqref{eq:OT_formula}, a function in the tangent space $\mathcal{TG}_{\mu_\ast}$ has the following form
\begin{equation}
    t(T_{\mu_\ast}^\mu - \text{id})(x)
=
t(m - S(\Sigma_*, \Sigma) m_\ast) + t(S(\Sigma_*, \Sigma) -I_d)x, \quad x \in \mathbb{R}^d. \label{eq:log_gaussian}
\end{equation}
This form implies that the function space $\mathcal{TG}_{\mu_\ast}$ is a set of affine functions of $x \in \mathbb{R}^d$.
Note that $\text{Exp}_{\mu_\ast}g \in \mathcal{G}(\mathbb{R}^d)$ holds for any $g \in \mathcal{TG}_{\mu_\ast}$, and also $\text{Log}_{\mu_\ast}\mu \in \mathcal{TG}_{\mu_\ast}$ holds for any $\mu \in \mathcal{G}(\mathbb{R}^d)$.

\section{Model} \label{sec:model}

In this section, we define regression models between Gaussian spaces using the above notion of tangent spaces. 
We first present our key idea of modeling and then develop two models.

\subsection{Idea: Nearly isometry between Gaussian Space and Linear Matrix Space}
\label{sec:transformation}

As our key idea, we give a nearly isometric map from  Gaussian space $\mathcal{G}(\mathbb{R}^d)$ to a linear matrix space. 
For $d \ge 1$, we define a set of symmetric matrices as 
\begin{align*}
    \Xi_d = \{(a, V) \in \mathbb{R}^{d \times (d+1)}: a \in \mathbb{R}^d, V \in \text{Sym}(d)\}, 
\end{align*}
which is obviously a linear space.
We will give a map from $\mathcal{G}(\mathbb{R}^d)$ to $\Xi_d$ and show that this map has certain isometric properties. 
This isometry map plays a
critical role in our regression model, given in the next subsection.
We fix a non-singular Gaussian measure $\mu_\ast = N(m_\ast, \Sigma_\ast) \in \mathcal{G}(\mathbb{R}^d)$ as a reference measure.

Preliminarily, we introduce an inner product on the space $\Xi_d$. 
For $(a,V), (b, U) \in \Xi_d$, we define 
\begin{equation}
    \langle (a, V), (b, U) \rangle_{m_\ast, \Sigma_\ast}
    =
    (a + V m_\ast)^\top(b + U m_\ast)
    +
    \text{tr}(V\Sigma_\ast U).
    \label{eq:inner_product}
\end{equation}
Then we can easily check that $\langle \cdot, \cdot \rangle_{m_\ast, \Sigma_\ast}$ satisfies the conditions of inner product.
This design follows an inner product for a space of affine functions.
Rigorously, for $a \in \mathbb{R}^d$ and $V \in \text{Sym}(d)$, we define an affine function $f_{a, V}(x) = a + Vx$ and its space $\mathcal{F}_{\text{aff}} = \{f_{a, V}: a \in \mathbb{R}^d, V \in \text{Sym}(d)\}$. 
Note that $\mathcal{TG}_{{\mu_\ast}} \subset \mathcal{F}_{\text{aff}}$ holds from \eqref{eq:log_gaussian}.
Then we consider an inner product between $f_{a,V}, f_{b,U} \in \mathcal{F}_{\text{aff}}$ with $(a,V), (b,U) \in \Xi_d$ as
\[
\langle f_{a,V}, f_{b,U}\rangle_{\mu_\ast}
=
\int_{\mathbb{R}^d}(a+Vx)^\top (b+Ux)d\mu_\ast(x)
=
(a + V m_\ast)^\top(b + U m_\ast)
    +
    \text{tr}(V\Sigma_\ast U).
\]
Inspired by the design, we obtain an inner product space $(\Xi_d, \langle \cdot , \cdot \rangle_{(m_\ast, \Sigma_\ast)})$.
The norm $\|\cdot\|_{(m_\ast, \Sigma_\ast)}$ induced by this inner product is specified as  
\begin{equation}
    \|(a, V)\|_{(m_\ast, \Sigma_\ast)}
=
\sqrt{\|a + Vm_\ast\|^2 + \text{tr}(V\Sigma_\ast V)}.
\label{eq:norm}
\end{equation}

We construct a nearly isometric map $\varphi_{\mu_\ast}$ from $(\mathcal{G}(\mathbb{R}^d), d_W)$ to $(\Xi_d, \|\cdot\|_{(m_\ast, \Sigma_\ast)})$ as
\begin{align}
    \varphi_{\mu_\ast} = \pi \circ \psi_{\mu_\ast}. \label{def:varphi}
\end{align}
We specify the maps $\psi_{\mu_\ast}: \mathcal{G}(\mathbb{R}^d) \to \mathcal{TG}_{{\mu_\ast}}$ and $\pi: \mathcal{F}_{\text{aff}} \to \Xi_d$ as follows.
First, $\psi_{\mu_\ast}$ is the logarithm map $\text{Log}_{\mu_\ast}(\cdot)$ as \eqref{eq:log_map} with restriction to $\mathcal{G}(\mathbb{R}^d)$. 
That is, for $\mu = N(m, \Sigma) \in \mathcal{G}(\mathbb{R}^d)$, $\psi_{\mu_\ast}\mu$ is the affine function of the form \eqref{eq:log_gaussian}.
Second, for an affine function $f_{a,V} \in \mathcal{F}_{\text{aff}}$, we define 
$
\pi f_{a,V}
=
(a, V).
$
For summary, the map $\varphi_{\mu_\ast}: \mathcal{G}(\mathbb{R}^d) \to \Xi_d$ in \eqref{def:varphi} is specified as
\begin{equation}
    \varphi_{\mu_\ast}\mu 
=
(m-S(\Sigma_\ast, \Sigma)m_\ast, S(\Sigma_\ast, \Sigma) - I), \quad \mu = N(m, \Sigma) \in \mathcal{G}(\mathbb{R}^d).
\label{eq:phi_map}
\end{equation}
We also define a map $\xi_{\mu_\ast}: \varphi_{\mu_\ast}\mathcal{G}(\mathbb{R}^d) \to \mathcal{G}(\mathbb{R}^d)$ as the left inverse of the map $\varphi_{\mu_\ast}$ by 
\begin{equation}
    \xi_{\mu_\ast}(a, V)
=
N(a+(V+I)m_\ast, (V+I)\Sigma_\ast(V+I)), \quad (a, V) \in \varphi_{\mu_\ast}\mathcal{G}(\mathbb{R}^d).
\label{eq:xi_map}
\end{equation}
Here,
a range of the map \eqref{eq:phi_map} with the domain $\mathcal{G}(\mathbb{R}^d)$ is written as
\begin{equation}
\varphi_{\mu_\ast}\mathcal{G}(\mathbb{R}^d)
=
\{(a, V) \in \Xi_d: V + I_d
\,\ \text{is positive semidefinite}\},
\end{equation}
which is obviously a subset of $\Xi_d$.

We obtain results on the distance-preserving property of the map $\varphi_{\mu_\ast}$.
As a preparation, for a $d \times d$ orthogonal matrix $U$, we define a class of Gaussian measures $\mathscr{C}_U \subset \mathcal{G}(\mathbb{R}^d)$ as
\begin{equation*}
    \mathscr{C}_U =
    \{ N(m, \Sigma) \in \mathcal{G}(\mathbb{R}^d) : m \in \mathbb{R}^d, \,\ U\Sigma U^\top \text{is diagonal} \}. 
\end{equation*}
Here, we give a formal statement.

\begin{prp}
Let $\mu_\ast \in \mathcal{G}(\mathbb{R}^d)$ be an arbitrary fixed reference measure. 
For any $\mu \in \mathcal{G}(\mathbb{R}^d)$,
we have
\begin{equation*}
    d_W(\mu, \mu_\ast)
    =
\|\varphi_{\mu_\ast}\mu\|_{(m_\ast, \Sigma_\ast)}.
\end{equation*}
Moreover, if $\mu_\ast \in \mathscr{C}_U$ holds, we have the following for any $\mu_1, \mu_2 \in \mathscr{C}_U$:
\[
d_W(\mu_1, \mu_2)
=
\|\varphi_{\mu_\ast}\mu_1 - \varphi_{\mu_\ast}\mu_2 \|_{(m_\ast, \Sigma_\ast)}.
\]
\label{prop:isomet}
\end{prp}

Note that since $\varphi_{\mu_\ast}\mu_\ast = 0$ holds, 
the first claim shows that
the Wasserstein distance between any Gaussian measure $\mu$ and the reference Gaussian measure $\mu_\ast$ is
equal to the distance between corresponding 
elements in the space $(\Xi_d, \|\cdot\|_{(m_\ast, \Sigma_\ast)})$.
The second claim shows that 
if we choose a class of Gaussian measures appropriately, the map $\varphi_{\mu_\ast}$ is isometric on that class. This isometric property is essentially illustrated in Section 2.3.2 in \cite{panaretos2020invitation} for the case of centered Gaussian distributions. Our claim can be understood as its generalization to the non-centered case.

\subsection{Regression Model}

In this section, 
we develop our regression models for the Gaussian-to-Gaussian distribution regression.
Our strategy is to map Gaussian distributions to the linear matrix spaces using the nearly isometric maps and then conduct linear regression between the matrix spaces.
Figure \ref{fig:regression_model} illustrates the strategy.
Specifically, we develop the following two models: (i) a basic model, and (ii) a low-rank model.
See Section \ref{subsec:notation} for the notation regarding matrices and tensors.

\begin{figure}
    \centering
    \includegraphics[width=0.9\hsize]{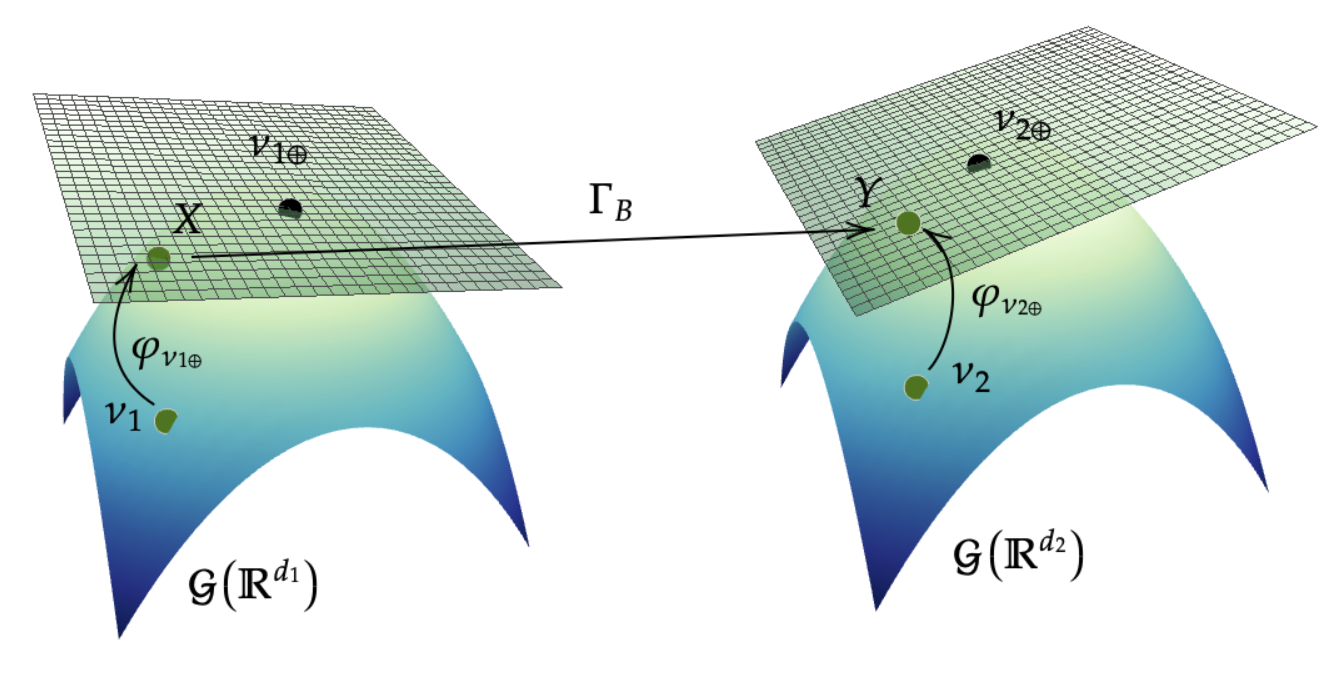}
    \caption{Illustration of structure of the proposed regression model between the Gaussain spaces $\mathcal{G}(\mathbb{R}^{d_1})$ and $\mathcal{G}(\mathbb{R}^{d_2})$. The Gaussian distributions $\nu_1$ and $\nu_2$ are transformed to the random elements $X$ and $Y$ in the linear matrix spaces $\Xi_{d_1}$ and $\Xi_{d_2}$ by the nearly isometric maps $\varphi_{\nu_{1\oplus}}$ and $\varphi_{\nu_{2\oplus}}$, respectively. Then, linear regression model with regression map $\Gamma_{\mathbb{B}_0}$ is assumed between $X$ and $Y$.}
    \label{fig:regression_model}
\end{figure}

We review the setup of the regression problem.
Let $d_1$ and $d_2$ be positive integers and $\mathcal{F}$ be a joint distribution on $\mathcal{G}(\mathbb{R}^{d_1}) \times  \mathcal{G}(\mathbb{R}^{d_2})$. 
Let $(\nu_1, \nu_2)$ be a pair of random elements generated by $\mathcal{F}$, where
we write $\nu_1 = N(m_1, \Sigma_1)$ and $ \nu_2 = N(m_2, \Sigma_2)$. 
We assume $\nu_1$ and $\nu_2$ are square integrable in the sense that $\max\{\mathbb{E}[d_W^2(\mu_1, \nu_{1})],\mathbb{E}[d_W^2(\mu_2, \nu_{2})]  \}< \infty$ for some (and thus for all) $\mu_1 \in \mathcal{G}(\mathbb{R}^{d_1})$ and 
$\mu_2 \in \mathcal{G}(\mathbb{R}^{d_2})$. 
In the following, we give models for dealing with this joint distribution $\mathcal{F}$.

\subsubsection{Basic model}

The first step is to define reference measures to introduce the nearly isometric maps.
For $j \in \{1,2\}$, we define the Fr\'{e}chet mean of the random Gaussian distribution $\nu_j$ as 
\begin{equation}
    \nu_{j\oplus} = N(m_{j\oplus}, \Sigma_{j\oplus})
    =
    \argmin_{\mu_j \in \mathcal{G}(\mathbb{R}^{d_j})} \mathbb{E}[d_W^2(\mu_j, \nu_j)],
    \label{eq:frechet}
\end{equation}
with the mean vector $m_{j\oplus} \in \mathbb{R}^{d_j}$ and the covariance matrix $\Sigma_{j\oplus} \in \mathbb{R}^{d_j \times d_j}$.
Note that the Fr\'echet means $\nu_{1\oplus}$ and $\nu_{2\oplus}$ are also Gaussian, and we assume they uniquely exist and are non-singular.

Using the Fr\'{e}chet means $\nu_{1 \oplus}$ and $\nu_{2 \oplus}$ as reference measures,
we transform random Gaussian distributions $\nu_1$ and $\nu_2$ to 
random elements $X \in \Xi_{d_1}$ and $Y \in \Xi_{d_2}$ by 
\begin{align*}
    X = \varphi_{\nu_{1\oplus}} \nu_1, \mbox{~and~}Y = \varphi_{\nu_{2\oplus}} \nu_2,
\end{align*}
where $\varphi_{\nu_{1\oplus}} $ and $\varphi_{\nu_{2\oplus}} $ are the nearly isometric maps in \eqref{eq:phi_map}.

For the random matrices $X$ and $Y$ transformed from the random distributions $\nu_1$ and $\nu_2$ as above, we perform a matrix-to-matrix linear regression.
To the aim, we consider a coefficient tensor $\mathbb{B} \in  \mathbb{R}^{d_1 \times(d_1+1)\times d_2 \times(d_2+1)}$ and define its associated linear map 
\begin{align*}
    \Gamma_{\mathbb{B}}: \mathbb{R}^{d_1 \times (d_1+1)} \to \mathbb{R}^{d_2 \times (d_2 + 1)}, { B} \mapsto \langle {B}, \mathbb{B} \rangle_2.
\end{align*}
Remind that $\langle \cdot, \cdot \rangle_2$ is a product for tensors defined in Section \ref{subsec:notation}.
To deal with the symmetricity of matrices in $\Xi_{d_1}$ and $\Xi_{d_2}$, we define the following class of coefficient tensors:
\begin{align}
    \mathcal{B} = 
    \{&
    \mathbb{B} \in \mathbb{R}^{d_1 \times(d_1+1)\times d_2 \times(d_2+1)} \\
    &:  
    \mathbb{B}[\cdot, \cdot, r, s] = \mathbb{B}[\cdot, \cdot, s-1, r+1] \,\ \text{for}\,\ 1 \le r \le d_2, 2 \le s \le d_2+1 
    \}.
    \label{eq:space_basic}
\end{align}
This definition guarantees $\langle {B}, \mathbb{B} \rangle_2 \in \Xi_{d_2}$ holds for any $\mathbb{B} \in \mathcal{B}$ and ${B} \in \Xi_{d_1}$.

We now give the linear regression model.
We assume that the $(\Xi_{d_1} \times \Xi_{d_2})$-valued random element $(X,Y)$, which is obtained by the transform of the random pair if distributions $(\nu_1,\nu_2)$, follows the following linear model with some $\mathbb{B}_0 \in \mathcal{B}$:
\begin{align}
    Y = \Gamma_{\mathbb{B}_0} (X) + E,  ~~ \mathbb{E}[E | X] = 0, \label{eq:model}
\end{align}
where $E$ is a $\Xi_{d_2}$-valued random element as an error term.
Note that $\mathbb{B}_0$ is not necessarily unique.
We can rewrite this model into an element-wise representation such that 
\begin{align}
    Y[r, s] = \langle X, \mathbb{B}_0[\cdot, \cdot, r, s] \rangle + E[r, s], \quad \mathbb{E}[E[r, s] | X] = 0, \label{eq:model_2}
\end{align} 
for $ 1 \le r \le d_2, 2 \le s \le d_2+1$.
Furthermore, we impose the following assumption on the data-generating process in this model:
\begin{equation}
    \Gamma_{\mathbb{B}_0}(X) \in \varphi_{\nu_{2\oplus}}\mathcal{G}(\mathbb{R}^{d_2}) \quad \text{with probability 1}.
    \label{eq:log_cond}
\end{equation}

For summary, we consider a regression map $\Gamma_{\mathcal{G}, \mathbb{B}_0}$ between the Gaussian spaces $\mathcal{G}(\mathbb{R}^{d_1})$ and $\mathcal{G}(\mathbb{R}^{d_2})$ as
\begin{align}
    \Gamma_{\mathcal{G}, \mathbb{B}_0} = \xi_{\nu_{2\oplus}} \circ \Gamma_{\mathbb{B}_0} \circ \varphi_{\nu_{1\oplus}}. \label{def:regression_model}
\end{align}
Note that our model satisfies $\Gamma_{\mathcal{G}, \mathbb{B}}(\nu_{1\oplus}) = \nu_{2\oplus}$ for any $\mathbb{B}$, since we have $\varphi_{\nu_{1\oplus}}\nu_{1\oplus} = 0$ and $\xi_{\nu_{2\oplus}}(0) = \nu_{2\oplus}$, 

Note that our model satisfies $\Gamma_{\mathcal{G}, \mathbb{B}_0}(\nu_{1\oplus}) = \nu_{2\oplus}$ , since we have $\varphi_{\nu_{1\oplus}}\nu_{1\oplus} = 0$ and $\xi_{\nu_{2\oplus}}(0) = \nu_{2\oplus}$. In other words, the regression map $\Gamma_{\mathcal{G}, \mathbb{B}_0}$ maps the Fr\'{e}chet mean of $\nu_{1}$ to that of $\nu_{2}$.

\begin{rmk}[Scalar response model]
A variant of the proposed basic model is the pairing of Gaussian distributions with scalar responses. 
In this case, the regression comes down to matrix-to-scalar linear regression.
Let  $(\nu_1, Z)$ be a pair of random elements
with a joint distribution on $\mathcal{G}(\mathbb{R}^{d_1}) \times \mathbb{R}$, and
let $\nu_{1\oplus} = (m_{1\oplus}, \Sigma_{1\oplus})$ be the Fr\'{e}chet mean of $\nu_1$ in $\mathcal{G}(\mathbb{R}^{d_1})$. 
A  Gaussian distribution-to-scalar regression model is 
\begin{align}
    Z = \langle X, \mathbb{B}_0 \rangle + \varepsilon, \quad \mathbb{E}[\varepsilon | X] = 0. \label{def:linear_model}
\end{align}
Here,  $X = \varphi_{\nu_{1\oplus}}\nu_1$  is an element in $\Xi_{d_1}$, $\mathbb{B}_0 \in \mathbb{R}^{d_1 \times (d_1+1)}$ is the regression parameter and $\varepsilon$ is a real-valued error term.
\end{rmk}

\subsubsection{Low-Rank Model}
\label{sec:low_rank_model}

We consider the case where the coefficient tensor $\mathbb{B}$ is assumed to have low-rank, as an extension of the basic model.
The issue with the basic model \eqref{eq:model} is that 
the number of elements in $\mathbb{B}$ is $d_1(d_1+1)d_2(d_2+1)$, which is high dimensional and far exceeds the usual sample size when $d_1$ and $d_2$ are not small.
A natural way to handle this issue is to approximate $\mathbb{B}$ with fewer parameters, and we
employ the low-rank CP decomposition of tensors for that purpose.
This approach was employed by \cite{zhou2013tensor} for a tensor regression model for scalar outcome, and by \cite{lock2018tensor} for a tensor-on-tensor regression model.

We define the low-rank coefficient tensor.
Let $K$ be a positive integer such that $K\le \min \{d_1, d_2\}$.
Then a tensor $\mathbb{A} \in \mathbb{R}^{d_1 \times (d_1+1) \times d_2 \times (d_2+1)}$ admits a rank-$K$ decomposition (e.g., \cite{kolda2009tensor}), if it holds that
\begin{align}
    \mathbb{A} = 
    \sum_{k=1}^K a_1^{(k)} \circ a_2^{(k)} \circ a_3^{(k)} \circ a_4^{(k)},
\label{eq:decomposition}
\end{align}
where
 $a_1^{(k)} \in \mathbb{R}^{d_1},  a_2^{(k)} \in \mathbb{R}^{d_1+1}, a_3^{(k)} \in \mathbb{R}^{d_2}, a_4^{(k)} \in \mathbb{R}^{d_2+1} (k=1, ..., K)$ are all column vectors. 
For convenience, we represent the decomposition \eqref{eq:decomposition} by a shorthand
\begin{align}
\mathbb{A} = \llbracket A_1, A_2, A_3, A_4 \rrbracket,    
\label{eq:decomposition_short}
\end{align}
 where $A_1 = [a_1^{(1)}, ..., a_1^{(K)}] \in \mathbb{R}^{d_1 \times K}, A_2 = [a_2^{(1)}, ..., a_2^{(K)}] \in \mathbb{R}^{(d_1+1) \times K}, A_3 = [a_3^{(1)}, ..., a_3^{(K)}] \in \mathbb{R}^{d_2 \times K}, A_4 = [a_4^{(1)}, ..., a_4^{(K)}] \in \mathbb{R}^{(d_2+1) \times K}$.

Based on this decomposition, we propose a rank-$K$ model for Gaussian distribution-to-distribution regression.
We will use the following notations: for a matrix $C \in \mathbb{R}^{d_2 \times (d_2+1)}$, we define a matrix $C^\ast \in \mathbb{R}^{d_2 \times (d_2+1)}$ by $C^\ast[r, 1] = C[r, 1]$ for $1 \le r \le d$ and $C^\ast[r, s] = C[s-1, r+1]$ for $1 \le r \le d_2, 2 \le s \le d_2+1$. Moreover, for a tensor $\mathbb{A} \in \mathbb{R}^{d_1 \times (d_1+1) \times d_2 \times (d_2+1)}$,  we define a tensor   $\mathbb{A}^\ast \in \mathbb{R}^{d_1 \times (d_1+1) \times d_2 \times (d_2+1)}$ as $\mathbb{A}^\ast[p, q, \cdot, \cdot]  = \mathbb{A}[p, q, \cdot, \cdot]^\ast$ for $1 \le p \le d_1, 1\le q \le d_1+1$. 
Then, we consider the regression parameter $\mathbb{B}_0$ in \eqref{eq:model} is assumed to have the form $\mathbb{B}_0 = (\mathbb{A}_0 + \mathbb{A}_0^\ast)/2$, where $\mathbb{A}_0 \in \mathbb{R}^{d_1 \times (d_1+1) \times d_2\times(d_2+1)}$ is a tensor with the rank-$K$ decomposition \eqref{eq:decomposition}. 
Under this assumption, the symmetric condition in \eqref{eq:space_basic} holds, that is, we have 
\[
\langle B, \mathbb{B}_0 \rangle_2
=
\left\langle B, \frac{\mathbb{A}_0 + \mathbb{A}_0^\ast}{2} \right\rangle_2 
=
\frac{\langle B, \mathbb{A}_0 \rangle_2 + \langle B, \mathbb{A}_0^\ast \rangle_2}{2}
=
\frac{\langle B, \mathbb{A}_0 \rangle_2 + \langle B, \mathbb{A}_0 \rangle_2^\ast}{2} \in \Xi_{d_2}
\]
for any $B \in \mathbb{R}^{d_1 \times (d_1+1)}$.

We denote the 
resulting parameter space for the rank-$K$ model as
 \begin{align}
 \mathcal{B}_{\text{low}} = \{\mathbb{B} = &(\mathbb{A} + \mathbb{A}^\ast)/2 \in \mathbb{R}^{d_1 \times (d_1+1) \times d_2 \times(d_2+1)} \notag\\ 
  &:\text{$\mathbb{A}$ has the rank-$K$ decomposition \eqref{eq:decomposition} } \}. 
   \label{eq:space_low}
\end{align}
The number of elements of the tensor $\mathbb{B} \in \mathcal{B}_{\text{low}}$ is $2K(d_1+d_2+1)$, which is much smaller than $d_1(d_1+1)d_2(d_2+1)$ when $d_1$ and $d_2$ are large.

\subsection{
Comparison with Existing Models in Terms of Generalization to Multivariate Case}\label{sec:comarison}

For the univariate case where $d_1 = d_2 = 1$, regression models applying the Wasserstein metric to distribution-on-distribution were introduced by \cite{chen2021wasserstein,zhang2022wasserstein, ghodratidistribution, ghodrati2023transportation, zhu2023autoregressive}.
\cite{chen2021wasserstein} and \cite{zhang2022wasserstein} transformed distributions in the Wasserstein space $\mathcal{W}(\mathbb{R})$ to elements in  its tangent space \eqref{eq:tanent_space} by the logarithmic map \eqref{eq:log_map}, and boiled down distribution-on-distribution regression to function-on-function linear regression.
Because the logarithmic map \eqref{eq:log_map} is isometric in the univariate case, their methods fully utilize the geometric properties of the Wasserstein space.
\cite{ghodratidistribution} modeled the regression operator from $\mathcal{W}(\mathbb{R})$ to $\mathcal{W}(\mathbb{R})$ by using the optimal transport map.
This approach enabled to interpret the regression effect directly at the level of probability distributions through a re-arrangement of probability mass.

Despite the effectiveness of these models for univariate distribution-on-distribution regression, their extension to the multivariate scenario remains non-trivial.
This challenge primarily arises from two reasons.
The first reason is that the explicit solution of the optimal transport problem for univariate distributions \eqref{eq:ot_onedim} is not available for the multivariate case. 
This brings numerical difficulties in the computation of optimal transport maps, which is required to transform distributions to unconstrained functions in the model by \cite{chen2021wasserstein}. 
The derivation of optimal transport maps also becomes essential when devising estimators for the regression map within \cite{ghodratidistribution}'s model.
The second reason is that the flatness of the Wasserstein space, that is, the isometric property of the logarithmic map \eqref{eq:isomet_onedim}, does not hold for the multivariate case in general.
This means the transformation method by \cite{chen2021wasserstein} lacks the theoretical support for preserving the geometric properties of the Wasserstein space in the multivariate case. 
Moreover, the identifiability result of the regression map in the model by \cite{ghodratidistribution}, which depends on the flatness of the Wasserstein space, is hard to be generalized for the multivariate case.
Another study \cite{ghodrati2023transportation} analyzes the multivariate case and reveals several theoretical properties such as the sample complexity.

We addressed these challenges by limiting the class of distributions to Gaussian distributions.
In our model, we transform Gaussian distributions to unconstrained matrices via the map \eqref{eq:phi_map}. 
Consequently, we simplify the regression of Gaussian distribution-on-Gaussian distribution to matrix-on-matrix linear regression. Given the explicit expression of the optimal transport map between Gaussian distributions as \eqref{eq:OT_formula}, our transformation avoids computational difficulties.
Although our transformation is not isometric in general, it has certain isometric properties as shown in Proposition \ref{prop:isomet}.
This guarantees that our transformation method partially utilizes the geometric properties of the Gaussian space.

\subsection{Generalization to Elliptically Symmetric Distributions}

Our proposed regression models extend to scenarios where distributions $\nu_1$ and $\nu_2$ belong to the class of elliptically symmetric distributions, a broader category than Gaussian distributions.
This is because, as shown in \cite{gelbrich1990formula}, the closed-form expression of the Wasserstein distance \eqref{eq:W_formula} holds if two distributions are in the same class of elliptically symmetric distributions.

We give more rigorous description.
Let $d \ge 1$ and let $f: [0, \infty) \to [0, \infty)$ be a measurable function that is not almost everywhere zero and satisfies
\begin{align}
    \int_{-\infty}^\infty |t|^\ell f(t^2)dt < \infty, \quad \ell = d-1, d, d+1.
    \label{f_cond}
\end{align}
Given such a function $f$, for a positive definite matrix $A \in \mathbb{R}^{d \times d}$ and a vector $v \in \mathbb{R}^d$,  one can consider a density function of the form $f_{A, v}(x) = (c_A)^{-1}f((x-v)^\top A(x-v)), x \in \mathbb{R}^d$. Here, 
we define $c_A = \int_{\mathbb{R}^d} f((x-v)^\top A(x-v)) dx$ as the normalizing constant.
Then, we can consider a class of distributions on $\mathbb{R}^d$ whose elements have a density $f_{A, v}$ for some positive definite matrix $A \in \mathbb{R}^{d \times d}$ and vector $v \in \mathbb{R}^d$.
We denote such a class as $\mathcal{P}_f(\mathbb{R}^d)$, and call it as the class of 
elliptically symmetric distributions with function $f$.
For example, if we set $f(t) = e^{-t/2}$, we obtain the set of Gaussian distributions with positive definite covariance matrices as $\mathcal{P}_f(\mathbb{R}^d)$.
Furthermore, by setting $f(t) = I_{[0, 1]}(t)$, we obtain the set of uniform distributions on ellipsoids of the forms $U_{A, v} = \{x \in \mathbb{R}^d: (x-v)^{\top}A(x-v) \le 1\}$ for some positive definite matrix $A \in \mathbb{R}^{d \times d}$ and vector $v \in \mathbb{R}^d$.

According to Theorem 2.4 of \cite{gelbrich1990formula}, the closed-forms of the Wasserstein distance \eqref{eq:W_formula} and optimal transport map \eqref{eq:OT_formula} are valid for any two measures $\mu_1, \mu_2$ in the same class of elliptically symmetric distributions $\mathcal{P}_f(\mathbb{R}^d)$.
Since our models rely only the forms \eqref{eq:W_formula}, \eqref{eq:OT_formula}, our result can be extended to the case in which $(\nu_1, \nu_2)$ are $\mathcal{P}_{f_1}(\mathbb{R}^{d_1}) \times \mathcal{P}_{f_2}(\mathbb{R}^{d_2})$-valued random elements. 
Note that $f_1,f_2:[0, \infty)\to [0, \infty)$ should be non-vanishing and satisfy the condition \eqref{f_cond}  for $d=d_1$ and $d=d_2$, respectively.

\section{Empirical Risk Minimization Algorithms} \label{sec:estimator}
In this section, we propose empirical risk minimization procedures for constructing a prediction model following the regression map $\Gamma_{\mathcal{G}, \mathbb{B}_0}$ (\ref{def:regression_model}) based on observed data. 
Specifically, we consider two cases: (i) we directly observe random distributions (Section \ref{sec:est_full}), and (ii) we observe only samples from the random distributions (Section \ref{sec:est_notfull}).
We refer the estimation issue of the coefficient tensor $\mathbb{B}_0$ itself and its related topics to Appendix.

\subsection{Algorithm with Directly Observed Distributions}
\label{sec:est_full}
Suppose that we directly observe $n$ independent pairs of random Gaussian distributions $(\nu_{1i}, \nu_{2i}) \sim \mathcal{F}$ for $i=1,...,n$.
Here, we write $\nu_{ji} = N(\mu_{ji}, \Sigma_{ji})$ for $j \in \{1, 2\}$.
Firstly, based on the distributions $\nu_{ji}(i=1, ..., n;j=1, 2)$, we compute the empirical Fr\'{e}chet means for $j \in \{1,2\}$:
\begin{align}
    \tilde{\nu}_{j\oplus}
    =
    \argmin_{\mu_j \in \mathcal{G}(\mathbb{R}^{d_j})} \frac{1}{n} \sum_{i=1}^n d_W^2(\mu_j,{\nu}_{ji} ),
    \label{empirical_frechet}
\end{align}
where we write  $\tilde{\nu}_{j\oplus} = N(\tilde{m}_{j\oplus}, \tilde{\Sigma}_{j\oplus})$.
For solving optimizations in \eqref{empirical_frechet},
we can use
the steepest descent algorithm (Section 5.4.1 in \cite{panaretos2020invitation}).
Then, we transform Gaussian distributions ${\nu}_{ji}$ into matrices by $\tilde{X}_i = \varphi_{\tilde{\nu}_{1\oplus}}\nu_{1i}$ and 
$\tilde{Y}_i = \varphi_{\tilde{\nu}_{2\oplus}}\nu_{2i}$.
In the basic model, we solve the following least squares problem:
\begin{equation}
    \tilde{\mathbb{B}}
\in
\argmin_{\mathbb{B} \in \mathcal{B}} \sum_{i=1}^n \|\tilde{Y}_i - \Gamma_{\mathbb{B}}(\tilde{X}_i)\|^2_{(\tilde{m}_{2\oplus}, \tilde{\Sigma}_{2\oplus})},
\label{eq:optim_basic}
\end{equation}
where $\mathcal{B}$ is the parameter space defined by \eqref{eq:space_basic}, and $\| \cdot \|_{(\tilde{m}_{2\oplus}, \tilde{\Sigma}_{2\oplus})}$ denotes the norm defined by \eqref{eq:norm}
for $m_\ast = \tilde{m}_{2\oplus} $ and $\Sigma_\ast = \tilde{\Sigma}_{2\oplus}$.
In the rank-$K$ model, we solve the following least squares problem:
\begin{equation}
    \tilde{\mathbb{B}}
\in
\argmin_{\mathbb{B} \in \mathcal{B}_{\text{low}}} \sum_{i=1}^n \|\tilde{Y}_i - \Gamma_{\mathbb{B}}(\tilde{X}_i) \|_{(\tilde{m}_{2\oplus}, \tilde{\Sigma}_{2\oplus})}^2,
\label{eq:optim_low}
\end{equation}
where $\mathcal{B}_{\text{low}}$ is the parameter space defined by \eqref{eq:space_low}.
In either case, we use $\Gamma_{\mathcal{G}, \tilde{\mathbb{B}}} = \xi_{\tilde{\nu}_{2\oplus}} \circ \Gamma_{\tilde{\mathbb{B}}} \circ \varphi_{\tilde{\nu}_{1\oplus}}$ as the map for prediction.

We propose an algorithm for solving the optimization problem in \eqref{eq:optim_low}. 
We observe that although the tensor $\mathbb{A}$ in 
$\mathbb{B} = (\mathbb{A} + \mathbb{A}^\ast)/2$ with rank $K$-decomposition \eqref{eq:decomposition_short} is not linear in $(A_1, A_2, A_3, A_4)$ jointly, it is linear in $A_c$ individually for $c=1, 2, 3, 4$. 
This observation suggests a so-called block relaxation algorithm (\cite{de1994block}), which alternately updates $A_c, c=1, 2, 3, 4$, while keeping the other matrices fixed. 
This algorithm is employed in \cite{zhou2013tensor} for parameter estimation in a tensor regression model.
We denote the objective function in the optimization problem in \eqref{eq:optim_low} as
\begin{align}
    \ell(A_1, A_2, A_3, A_4)
=
\sum_{i=1}^n \|\tilde{Y}_i - \Gamma_{\llbracket A_1, A_2, A_3, A_4 \rrbracket}(\tilde{X}_i) \|_{(\tilde{m}_{2\oplus}, \tilde{\Sigma}_{2\oplus})}^2.
\label{eq:optimfunc_tensor}
\end{align}

Then the procedure for solving the optimization problem in \eqref{eq:optim_low} is summarized in Algorithm \ref{alg:block}.
First, we generate initialized matrices $A_1^{(0)}, A_2^{(0)}, A_3^{(0)}, A_4^{(0)}$, whose elements follow  the uniform distribution on some compact interval.
Then, with a number of iteration $T \in \mathbb{N}$, we generate a sequence $\{(A_1^{(t)}, A_2^{(t)}, A_3^{(t)}, A_4^{(t)})\}_{t=1}^T$ by the iterative update in Algorithm \ref{alg:block}.
As the block relaxation algorithm monotonically decreases the objective function \cite{de1994block}, and the function $\ell$ is bounded from below, the convergence of objective values $\ell(A_1^{(t)}, A_2^{(t)}, A_3^{(t)}, A_4^{(t)})$ is guaranteed. 

The algorithm should be run multiple times with different initializations to get a better minimum.

\begin{figure}[!t]
\begin{algorithm}[H]
    \caption{Block relaxation algorithm for minimizing \eqref{eq:optim_low}.}
    \label{alg:block}
    \begin{algorithmic}
    \STATE \textbf{Initialize}:$A_1^{(0)} \in \mathbb{R}^{d_1 \times K}, A_2^{(0)} \in \mathbb{R}^{(d_1+1) \times K}, A_3^{(0)} \in \mathbb{R}^{d_2 \times K}, A_4^{(0)} \in \mathbb{R}^{(d_2+1)\times K}$. 
    \FOR{$t=1,...,T$}
    \STATE  
$A_1^{(t+1)} = \argmin_{A_1} \ell(A_1, A_2^{(t)}, A_3^{(t)}, A_4^{(t)})$
\STATE 
$A_2^{(t+1)} = \argmin_{A_2} \ell(A_1^{(t+1)}, A_2, A_3^{(t)}, A_4^{(t)})$
\STATE 
$A_3^{(t+1)} = \argmin_{A_3} \ell(A_1^{(t+1)}, A_2^{(t+1)}, A_3, A_4^{(t)})$
\STATE $A_4^{(t+1)} = \argmin_{A_4} \ell(A_1^{(t+1)}, A_2^{(t+1)}, A_3^{(t+1)}, A_4)$

    \ENDFOR
    \end{algorithmic}
\end{algorithm}
\end{figure}

\subsection{Algorithm with Samples of Not Directly Observed Distributions}
\label{sec:est_notfull}
In this section, suppose that we observe only samples from the random Gaussians $(\nu_{1i}, \nu_{2i})$, instead of the direct observation on $(\nu_{1i}, \nu_{2i})$ in Section \ref{sec:est_full}. 
Rigorously, we assume the following two-step data generating process. 
First, $n$ independent pairs of Gaussian distributions $(\nu_{1i}, \nu_{2i}) \sim \mathcal{F}~(i=1, ..., n)$ are generated.
Next, the $N$ sample vectors $W_{jim} \sim \nu_{ji} (m=1, ..., N)$ are generated from the distributions, then we observe the sample vectors.
For each fixed $(i, j)$, the $W_{jim}$ are independent and identically distributed.

At the beginning, we develop a proxy for each Gaussian distribution $\nu_{ji} = N(\mu_{ji}, \Sigma_{ji})$.
For $i=1,...,n$ and $j \in \{1,2\}$, we consider the empirical mean and covariance of $W_{jim}$ as
\[
    \hat{\mu}_{ji}
    =
    \frac{1}{N}\sum_{m=1}^N W_{jim} \quad 
    \text{and} \quad \hat{\Sigma}_{ji}
    =
    \frac{1}{N}\sum_{m=1}^N (W_{jim} - \hat{\mu}_{ji})(W_{jim} - \hat{\mu}_{ji})^\top,
\]
for estimators of $\mu_{ji}$ and $\Sigma_{ji}$, respectively. 
We define $\hat{\nu}_{ji} = N(\hat{\mu}_{ji}, \hat{\Sigma}_{ji})$ and use it for a proxy of $\nu_{ji} = N(\mu_{ji}, \Sigma_{ji})$.
Based on these proxies, we compute the empirical Fr\'{e}chet means for $j \in \{1,2\}$:
\begin{align}
    \hat{\nu}_{j\oplus}
    =
    \argmin_{\mu_j \in \mathcal{G}(\mathbb{R}^{d_j})} \frac{1}{n} \sum_{i=1}^n d_W^2(\mu_j,\hat{\nu}_{ji} ),
    \label{empirical_frechet_notfull}
\end{align}
where we write $\hat{\nu}_{1\oplus} = N(\hat{m}_{1\oplus}, \hat{\Sigma}_{1\oplus}), 
\hat{\nu}_{2\oplus} = N(\hat{m}_{2\oplus}, \hat{\Sigma}_{2\oplus})
$.
As with the directly observed case,
we can use
the steepest descent algorithm 
for solving this optimization. 
Then, we transform Gaussian distributions $\hat{\nu}_{ji}$ into matrices by $\hat{X}_i = \varphi_{\hat{\nu}_{1\oplus}}\hat{\nu}_{1i}$ and 
$\hat{Y}_i = \varphi_{\hat{\nu}_{2\oplus}}\hat{\nu}_{2i}$.
In the basic model, we solve the following least squares problem:
\[
\hat{\mathbb{B}}
\in
\argmin_{\mathbb{B} \in \mathcal{B}} \sum_{i=1}^n \|\hat{Y}_i - \Gamma_{\mathbb{B}}(\hat{X}_i)\|^2_{(\hat{m}_{2\oplus}, \hat{\Sigma}_{2\oplus})},
\]
where $\| \cdot \|_{(\hat{m}_{2\oplus}, \hat{\Sigma}_{2\oplus})}$ denotes the norm defined by \eqref{eq:norm}
for $m_\ast = \hat{m}_{2\oplus} $ and $\Sigma_\ast = \hat{\Sigma}_{2\oplus}$.
In the rank-$K$ model, we solve the following least squares problem:
\begin{align}
    \hat{\mathbb{B}}
\in
\argmin_{\mathbb{B} \in \mathcal{B}_{\text{low}}} \sum_{i=1}^n \|\hat{Y}_i - \Gamma_{\mathbb{B}}(\hat{X}_i) \|_{(\hat{m}_{2\oplus}, \hat{\Sigma}_{2\oplus})}^2.
\label{eq:optim_tensor_sample}
\end{align}
In either case, we use $\Gamma_{\mathcal{G}, \hat{\mathbb{B}}} = \xi_{\hat{\nu}_{2\oplus}} \circ \Gamma_{\hat{\mathbb{B}}} \circ \varphi_{\hat{\nu}_{1\oplus}}$ as the prediction map. 
As with the directly observed case, we can use the block relaxation algorithm for solving the optimization \eqref{eq:optim_tensor_sample} by the similar manner of Algorithm \ref{alg:block}.

\section{Analysis of in-sample prediction error} \label{sec:prediction_error}
In this section, we analyze the prediction error of the proposed models and algorithms.
We especially focus on the in-sample prediction error measured on the observations, which is naturally extended to the out-sample prediction error.
Here, suppose that we directly observe the pairs of Gaussian distributions  $(\nu_{1i}, \nu_{2i}), i=1, ..., n$ from the model \eqref{eq:model} as the case in Section \ref{sec:est_full}.
For simplicity, we assume that the true values of Fr\'{e}chet means $\nu_{1\oplus}$ and $\nu_{2\oplus}$ are known.
In addition, we treat predictors $\{\nu_{1i} \}_{i=1}^n$ as fixed in this analysis.
Based on the sample $(\nu_{1i}, \nu_{2i}), i=1, ..., n$, we solve the following least squares problem for $\tilde{\mathcal{B}} = \mathcal{B}$ or $\tilde{\mathcal{B}} = \mathcal{B}_{\text{low}}$:
\begin{equation}
    \tilde{\mathbb{B}} 
\in
\argmin_{\mathbb{B} \in {\tilde{\mathcal{B}}}}
\sum_{i=1}^n \|Y_i - \Gamma_{\mathbb{B}}(X_i) \|_{({m}_{2\oplus}, {\Sigma}_{2\oplus})}^2,
\label{eq:estmator_pred}
\end{equation}
where $X_i = \varphi_{\nu_{1\oplus}}\nu_{1i}$ and 
$Y_i = \varphi_{\nu_{2\oplus}}\nu_{2i}$.
Then, we define the prediction map.
Moreover, under the assumption that $\Gamma_{\tilde{\mathbb{B}}}(X_i) \in \varphi_{\nu_{2\oplus}}\mathcal{G}(\mathbb{R}^{d_2}) (i=1, ..., n)$, we define the in-sample prediction error with the Wasserstein metric in terms of the empirical measure by
\begin{align}
    \mathcal{R}_n({\Gamma}_{\mathcal{G}, \tilde{\mathbb{B}}}, {\Gamma}_{\mathcal{G}, {\mathbb{B}}_0}) =
\sqrt{\frac{1}{n}\sum_{i=1}^n d_W^2({\Gamma}_{\mathcal{G}, \tilde{\mathbb{B}}}(\nu_{1i}), 
{\Gamma}_{\mathcal{G}, {\mathbb{B}}_0}(\nu_{1i}))},
\end{align}
which is an analogy of the empirical $L^2$-norm.
We also assume that the $\Xi_d$-valued random variable $E$ in the linear model \eqref{def:linear_model} is Gaussian, that is,
that is, for any $A \in \Xi_d$,  $\langle E, A \rangle_{m_\ast, \Sigma_\ast}$ is a real Gaussian random variable.

In the following, we measure the in-sample prediction error of the basic model in terms of the Wasserstein distance. 
Note that this is unique to our distribution-on-distribution regression problem, and deriving the convergence rate of in-sample prediction error under this setting is not a trivial problem. 

\begin{thm}[Basic Model]
Suppose that $(\nu_{1i}, \nu_{2i}) (i=1, ..., n)$ are pairs of Gaussian distributions generated from the basic model \eqref{eq:model}, and that error matrices $E_i \in \Xi_{d_2}$ are Gaussian with mean $0$ and covariance with trace $1$, that is, $\mathbb{E}[E_i] = 0$ and $\mathbb{E}[\|E_i\|^2_{m_{2\oplus}, \Sigma_{2\oplus}}] = 1$.
Let $\tilde{\mathbb{B}} \in \mathcal{B}$ be an solution of the optimization \eqref{eq:estmator_pred}, and assume that $\Gamma_{\tilde{\mathbb{B}}}(X_i) \in \varphi_{\nu_{2\oplus}}\mathcal{G}(\mathbb{R}^{d_2})$ holds for $i=1, ..., n$.
Then, we have 
\begin{align*}
\mathcal{R}_n({\Gamma}_{\mathcal{G}, \tilde{\mathbb{B}}}, {\Gamma}_{\mathcal{G}, {\mathbb{B}}_0})= O_P(d_1d_2 / \sqrt{n}),
\end{align*}
as $n \to \infty$.
\label{thm:conv_rate_basic}
\end{thm}

This result shows that that our method achieves optimal convergence rates.
That is, the convergence rates in Theorem \ref{thm:conv_rate_basic} achieve the parametric rate $n^{-1/2}$ regarding the sample size $n$. 
This rate comes from our parametric assumption of Gaussianity on distributions. 
In contrast, existing distribution-on-distribution regression models do not impose parametric assumptions, which results in slower convergence rates of estimators for regression parameters. 
For example, in the regression model proposed by \cite{chen2021wasserstein}, an estimator for the regression operator achieve the same rate as the minimax rate for function-to-function linear regression in a certain case (Theorem1 in \cite{chen2021wasserstein}), which is generally  slower than the parametric rate. 
In the regression model proposed by \cite{ghodratidistribution}, an estimator for the regression map achieve the rate $n^{-1/3}$ (Theorem 3.8 in \cite{ghodratidistribution}), which is slower than the parametric rate.

Next, we study the in-sample prediction error of the rank-$K$ model.
This analysis provides an effect of the number of ranks $K$, in addition to the results of the basic model in Theorem \ref{thm:conv_rate_basic}.
\begin{thm}[Rank-$K$ Model]
Suppose $(\nu_{1i}, \nu_{2i}) (i=1, ..., n)$ are pairs of Gaussian distributions generated from the rank-K model defined in Section \ref{sec:low_rank_model}, and that error matrices $E_i \in \Xi_{d_2}$ are Gaussian with mean $0$ and covariance with trace $1$. 
Let $\tilde{\mathbb{B}} \in \mathcal{B}_{\text{low}}$ be an solution of the optimization 
\eqref{eq:estmator_pred}, and assume that $\Gamma_{\tilde{\mathbb{B}}}(X_i) \in \varphi_{\nu_{2\oplus}}\mathcal{G}(\mathbb{R}^{d_2})$ holds for $i=1, ..., n$.
Then, we have 
\begin{align}
\mathcal{R}_n({\Gamma}_{\mathcal{G}, \tilde{\mathbb{B}}}, {\Gamma}_{\mathcal{G}, {\mathbb{B}}_0})= O_P(\sqrt{K(d_1+d_2)} / \sqrt{n}),
\end{align}
as $n \to \infty$.
\label{thm:conv_rate_low}
\end{thm}

Theorem \ref{thm:conv_rate_low} states an advantage of the low-rank model, in addition to the result that the model achieves the optimal parametric rate.
The constant part of the rate is $\sqrt{K(d_1+d_2)}$ in the rank-$K$ model while $d_1d_2$  in the basic model. This implies that when the dimensions of distributions $\nu_1, \nu_2$ are large, the regression map in the rank-$K$ model is better approximated than that in the basic model.

We add some discussion on the observations of distributions. 
Recall that we assume the true Fr\'{e}chet means $\nu_{1\oplus}, \nu_{2\oplus}$ are known, and distributions $(\nu_{1i}, \nu_{2i})$ are directly observed. Relaxing these assumptions presents additional challenges for theoretical analysis. Specifically, if we estimate the Fr\'{e}chet mean of $\nu_{2i}$ with the empirical Fr\'{e}chet mean $\tilde{\nu}_{2\oplus}$, we solve the least squares problem \eqref{eq:estmator_pred} by replacing $Y_i =  \log_{\nu_{2\oplus}}\nu_{2i}$ with $\tilde{Y}_i = \log_{\tilde{\nu}_{2\oplus}}\nu_{2i}$. Since $\tilde{Y}_1, ..., \tilde{Y}_n$ are not independent, the standard theory for analyzing the error of empirical risk minimization is not directly applicable in this setting. Moreover, if distributions are not directly observed and only samples from them are available, we need to tackle the discrepancy between the estimated distributions based on the sample and the actual distributions in the analysis. As for the estimation of the Fr\'{e}chet mean, \cite{le2022fast} derive the rates of convergence of empirical Fr\'{e}chet mean on the Gaussian space (Corollary 17 in \cite{le2022fast}), which may be helpful for further theoretical analysis. 

Finally, we prove the consistency and asymptotic normality of an estimator for identified regression parameters in the Appendix.

\section{Simulation Studies} \label{sec:simulation}
In this section, we investigate the predictive  performance of the proposed methods together with an alternative regression method through simulation studies. 
The purpose of these studies is to validate the usage of the proposed nearly isometric map for improving the accuracy in predicting distributions in terms of the Wasserstein metric.

As an alternative regression approach, we consider the following model between $\nu_{1i} \in \mathcal{G}(\mathbb{R}^{d_1})$ and $\nu_{2i} \in \mathcal{G}(\mathbb{R}^{d_2})$:
\begin{equation}
    W_i = \langle Z_i, \mathbb{D}_0 \rangle_2 + E_i, \quad \mathbb{E}[E_i | Z_i] = 0. 
    \label{eq:alt_model}
\end{equation}
Here, $Z_i = (m_{1i}, \Sigma_{1i}) \in S_{d_1}$ and
$W_i = (m_{2i}, \Sigma_{2i}) \in S_{d_2}$ are matrices obtained from Gaussian distributions $\nu_{1i} = N(m_{1i}, \Sigma_{1i})$ and $\nu_{2i} = N(m_{2i}, \Sigma_{2i})$, respectively.
$\mathbb{D}_0 \in \mathcal{B}$ is the regression parameter and $E_i \in S_{d_2}$ is the error matrix in this model.
Note that this alternative model does not consider the Wasserstein metric.

\subsection{Setting}
\label{sec:setting}
Setting $d_1 = d_2 = d$, we generate pairs of Gaussian distributions
$\{(\nu_{1i}, \nu_{2i})\}_{i=1}^n$
from a mixture of the proposed and alternative models as follows.
First, for $i=1, ..., n$, we independently generate binary random variable $C_i \in \{0, 1\}$ such that $\mathbb{P}(C_i = 0) = \mathbb{P}(C_i = 1) = 1/2$. Then, we generate a pair $(\nu_{1i}, \nu_{2i})$ form the proposed model if $C_i=0$, and from the alternative model if $C_i=1$. The way to generate a pair from each model is as follows.

\subsubsection{Generation form proposed model}
We firstly generate independent random variables $G_i^{(1)}, ..., G_i^{(d)} \sim N(0, 1)$, $H_i^{(1)}, ..., H_i^{(d)} \sim Exp(1)$ and set a  matrix $X_i \in S_d$ by 
\[
X_i = \begin{pmatrix}
G_i^{(1)} &H_i^{(1)} & & \text{\huge{0}} \\
\vdots  & & \ddots & \\
 G_i^{(d)}& \text{\huge{0}} & & H_i^{(d)}
\end{pmatrix}.
\]
Here, 
$Exp(1)$ is the exponential distribution with the rate parameter $1$.
Then we obtain a Gaussian distribution $\nu_{1i} = \xi_{\nu_{1\oplus}}X_i \in \mathcal{G}(\mathbb{R}^d)$, where
$\nu_{1\oplus}$ is the $d$-dimensional standard Gaussian distribution. Note that under this setting, the random distribution $\nu_{1i}$ has the Fr\'{e}chet mean $\nu_{1\oplus}$.
Next, we set the coefficient tensor  $\mathbb{B}_0 \in \mathbb{R}^{d \times (d+1) \times d \times (d+1)}$ as  
\[
\mathbb{B}_0[\cdot, \cdot, r, 1]
=
\begin{pmatrix}
1 & 0 & \cdots & 0 \\
\vdots  & & \ddots & \\
 1 & 0 & \cdots & 0
\end{pmatrix}, \quad 
\mathbb{B}_0[\cdot, \cdot, r, r+1]
=
\begin{pmatrix}
0 &(2d)^{-1} & & \text{\huge{0}} \\
\vdots  & & \ddots & \\
 0& \text{\huge{0}} & & (2d)^{-1}
\end{pmatrix}, 
\]

for $1 \le r \le d$,
and set the other elements to be zero. Additionally, we 
generate independent random variables $U_i^{(1)}, ..., U_i^{(d)} \sim N(0, 1), V_i^{(1)}, ..., V_i^{(d)} \sim  U(-1/2, 1/2)$ and 
set the error matrix $E_i \in S_d$ by 
\[
E_i = 
\begin{pmatrix}
U_i^{(1)} &V_i^{(1)} & & \text{\huge{0}} \\
\vdots  & & \ddots & \\
 U_i^{(d)}& \text{\huge{0}} & & V_i^{(d)}
\end{pmatrix}.
\]

Here, $U(-1/2, 1/2)$ is the uniform distribution on the interval $(-1/2, 1/2)$.
We set $Y_i = \langle X_i, \mathbb{B}_0 \rangle_2 + E_i$ and obtain a response Gaussian distribution $\nu_{2i} = \xi_{\nu_{2\oplus}}Y_i \in \mathcal{G}(\mathbb{R}^d)$, where $\nu_{2\oplus}$ is the $d$-dimensional standard Gaussian distribution.
Note that under this setting, the condition \eqref{eq:log_cond} holds and the random distribution $\nu_{2i}$ has the  Fr\'{e}chet mean $\nu_{2\oplus}$.

\subsubsection{Generation from alternative model}
We firstly generate independent random variables $G_i^{(1)}, ..., G_i^{(d)} \sim N(0, 1)$, $H_i^{(1)}, ..., H_i^{(d)} \sim Exp(1)$ and set a  matrix $Z_i \in S_d$ by 
\[
Z_i = \begin{pmatrix}
G_i^{(1)} &H_i^{(1)}+1 & & \text{\huge{0}} \\
\vdots  & & \ddots & \\
 G_i^{(d)}& \text{\huge{0}} & & H_i^{(d)}+1
\end{pmatrix}.
\]
Then, we obtain the Gaussian distribution $\nu_{1i} = N(m_{1i}, \Sigma_{1i})$ such that $Z_i = (m_{1i}, \Sigma_{1i})$. Next, we set the coefficient tensor  $\mathbb{D}_0 \in \mathbb{R}^{d \times (d+1) \times d \times (d+1)}$ as 
\[
\mathbb{D}_0[\cdot, \cdot, r, 1]
=
\begin{pmatrix}
1 & 0 & \cdots & 0 \\
\vdots  & & \ddots & \\
 1 & 0 & \cdots & 0
\end{pmatrix}, \quad 
\mathbb{D}_0[\cdot, \cdot, r, r+1]
=
\begin{pmatrix}
0 &(2d)^{-1} & & \text{\huge{0}} \\
\vdots  & & \ddots & \\
 0& \text{\huge{0}} & & (2d)^{-1}
\end{pmatrix}, 
\]

for $1 \le r \le d$,
and set the other elements to be zero. Additionally, we generate independent random variables $U_i^{(1)}, ..., U_i^{(d)} \sim N(0, 1), V_i^{(1)}, ..., V_i^{(d)} \sim  U(-1/2, 1/2)$ and 
set the error matrix $E_i \in S_d$ by 
\[
E_i = 
\begin{pmatrix}
U_i^{(1)} &V_i^{(1)} & & \text{\huge{0}} \\
\vdots  & & \ddots & \\
 U_i^{(d)}& \text{\huge{0}} & & V_i^{(d)}
\end{pmatrix}.
\]

We set $W_i = \langle Z_i, \mathbb{D}_0 \rangle_2 + E_i$ and obtain the response Gaussian distribution $\nu_{2i} = N(m_{2i}, \Sigma_{2i})$ such that $W_i = (m_{2i}, \Sigma_{2i})$.

From the above procedures, we have obtained pairs of Gaussian distributions $\{(\nu_{1i}, \nu_{2i})\}_{i=1}^n$. Finally, we draw $N$ independent sample vectors from each of the distributions $\{\nu_{1i}\}_{i=1}^n$ and $\{\nu_{2i}\}_{i=1}^n$.

\subsection{Performance Criterion}
For the proposed models, we construct estimators $\hat{\mathbb{B}}$ as described in Section \ref{sec:est_notfull}. 
For the alternative model \eqref{eq:alt_model}, we construct an estimator by solving the least square problem
\[
\hat{\mathbb{D}} \in \argmin_{\mathbb{D} \in \mathcal{B}}\sum_{i=1}^n \|\hat{W}_i - \langle \hat{Z}_i, \mathbb{D} \rangle_2 \|_F^2, 
\]
where $\hat{Z}_i = (\hat{m}_{1i}, \hat{\Sigma}_{1i})$ and
$\hat{W}_i = (\hat{m}_{2i}, \hat{\Sigma}_{2i})$.

To investigate the performance of the proposed and alternative methods, following simulations in \cite{chen2021wasserstein}, we generate 200 new predictors $\{\nu_{1i}\}_{i=n+1}^{n+200}$ in the way of Section \ref{sec:setting}
and compute the out-of-sample average Wasserstein discrepancy (AWD). For $i=n+1, ..., n+200$, we define the true response distribution $\nu_{2i}^\ast = N(m_{2i}^\ast, \Sigma_{2i}^\ast)$ by $\nu_{2i}^\ast = \xi_{\nu_{1\oplus}}\langle X_i, \mathbb{B}_0 \rangle_2$ if $C_i=0$, and by $(m_{2i}^\ast, \Sigma_{2i}^\ast) = \langle Z_i, \mathbb{D}_0\rangle_2$ if $C_i = 1$. Then, denoting the fitted response distributions by $\nu_{2i}^\#$, the out-of-sample AWD is given by  
\begin{equation}
\text{AWD}
=
\frac{1}{200}\sum_{i=n+1}^{n+200}d_W(\nu_{2i}^{\ast}, \nu_{2i}^{\#}).
\label{eq:awd}
\end{equation}

In the proposed model, 
when the fit of the response in the space $\Xi_{d_2}$ 
does not fall in the range of map $\varphi_{\hat{\nu}_{2\oplus}}$ , that is, 
\begin{equation}
\Gamma_{\hat{\mathbb{B}}}(X_i) \notin \varphi_{\hat{\nu}_{2\oplus}}\mathcal{G}(\mathbb{R}^{d_2}),
    \label{eq:fall}
\end{equation}
we need to modify the fit to calculate the fitted response distribution.
To handle this problem,
we use a boundary projection method similar to one proposed by \cite{chen2021wasserstein}. 
Specifically,  for $d \ge 1$, let $g_d: \mathbb{R}^{d \times (d+1)} \to \mathbb{R}^{d \times d}$ be the map such that $g((a, V)) = V$ for $(a, V) \in \mathbb{R}^{d \times (d+1)}$. 
If the event \eqref{eq:fall} happens, we 
calculate a constant $\eta_i$ such that 
\[
\eta_i = 
\max\{\eta \in [0, 1]: \eta (g_{d_2} \circ \Gamma_{\hat{\mathbb{B}}}(X_i)) + I_{d_2}\,\ \text{is positive semidefinite}\},
\]
and update the original fit by $\eta_i \Gamma_{\hat{\mathbb{B}}}(X_i)$.
Conceptually, we update the original fit by a projection onto the boundary of $\varphi_{\hat{\nu}_{2\oplus}}\mathcal{G}(\mathbb{R}^{d_2})$ along the line segment between the origin 0 and the fit $\Gamma_{\hat{\mathbb{B}}}(X_i)$.
In the alternative method, if $g_{d_2}(\langle X_i, \hat{\mathbb{D}} \rangle_2)$ is not positive semidefinite, we update $g_{d_2}(\langle X_i, \hat{\mathbb{D}} \rangle_2)$ by $\argmin_{C \in \text{Sym}^+(d_2)}\|C - g_{d_2}(\langle X_i, \hat{\mathbb{D}} \rangle_2)\|_F$.

\subsection{Results}
Firstly, we set $d = 2$ and consider four scenarios with $n \in \{50, 200\}$ and $N \in \{50, 500\}$. 
We simulate 500 runs for each $(n, m)$ pair, and for each Monte Carlo run, we compute the AWD \eqref{eq:awd} based on 200 new predictors.
The results of the proposed and alternative methods are summarized in the boxplots of Figure \ref{fig:awd_comparison}. 
In all four scenarios, the proposed method outperforms the alternative method. This result comes from the fact that the proposed method takes into account the geometry of the Wasserstein metric, while the alternative method does not. 
In this setting, the event \eqref{eq:fall} seldom happened even if the number of distributions $n$ is small.

\begin{figure}
    \centering
    \includegraphics[width=120mm]{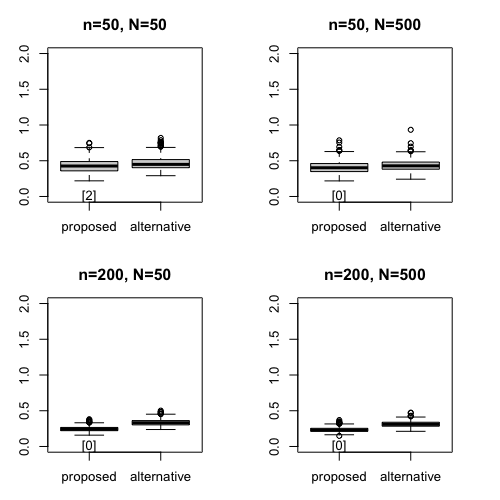}
    \caption{Boxplots of the out-of-sample AWDs defined as \eqref{eq:awd} for the four scenarios with $n \in \{50, 200\}$ and $N \in \{50, 500\}$. "proposed" denotes the proposed method and "alternative" denotes the alternative method. The number in brackets "[ ]" below the boxplots for the proposed indicates how many runs event \eqref{eq:awd} happened and boundary projection was needed.}
    \label{fig:awd_comparison}
\end{figure}

Next, we set $d=6, n=200, N=500$ and fit the proposed and alternative models whose regression tensors have rank $K \in \{2, 3, 4\}$. 
As with the previous experiment, we simulate 500 runs, and for each Monte Carlo run, we compute the AWD \eqref{eq:awd} based on 200 new predictors. 
The results are summarized in the boxplots of Figure \ref{fig:awd_comparison_tensor}.
In all cases, the proposed method outperforms the alternative method. In this setting, event (\ref{eq:fall}) happened more frequently than in the previous experiment.

\begin{figure}
    \centering
    \includegraphics[width=120mm]{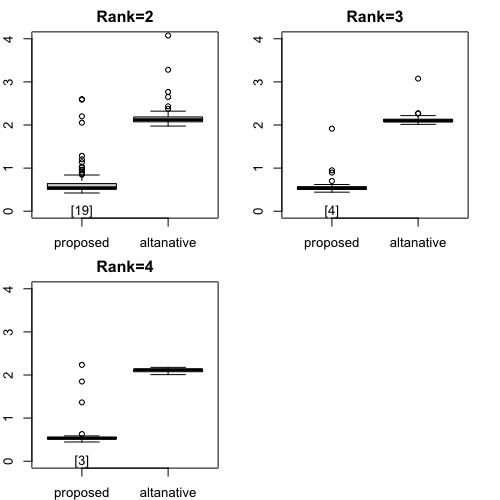}
    \caption{Boxplots of the out-of-sample AWDs defined as \eqref{eq:awd} for the low-rank methods with rank $K \in \{2, 3, 4\}$. "proposed" denotes the proposed method and "alternative" denotes the alternative method. The number in brackets "[ ]" below the boxplots for the proposed indicates how many runs event \eqref{eq:awd} happened and boundary projection was needed.}
    \label{fig:awd_comparison_tensor}
\end{figure}

Finally, to see the performance of the methods under the existence of model misspecification, we generate pairs of multivariate $t$-distributions $\{(t_{1i}, t_{2i})\}_{i=1}^n$ and fit the Gaussian-on-Gaussian regression models. 

Specifically, we firstly generate pairs of Gaussian distributions $\{(\nu_{1i}, \nu_{2i})\}_{i=1}^n$ from the mixture of the proposed and alternative models as described in Section \ref{sec:setting}.
Denoting these Gaussian distributions as $\nu_{1i} = N(m_{1i}, \Sigma_{1i}), \nu_{2i} = N(m_{2i}, \Sigma_{2i})$, we set multivariate $t$-distributions as $t_{1i} = t_\ell(m_{1i}, \Sigma_{1i}), t_{2i} = t_\ell(m_{2i}, \Sigma_{2i})$. 
Here, $t_{\ell}(m, \Sigma)$ denotes the multivariate t-distribution with location $m$, scale matrix $\Sigma$ and the degree of freedom $\ell$.
We draw an i.i.d. observations of size $N$ from each of the distributions $\{t_{1i}\}_{i=1}^n$ and $\{t_{2i}\}_{i=1}^n$, and construct estimators for the proposed and alternative models, respectively. Finally, we generate 200 new predictors $\{t_{1i}\}_{i=n+1}^{n+200}$ from the mixture model and calculate the out-of-sample AWD 
$
200^{-1}\sum_{i=n+1}^{n+200}d_W(t_{2i}^{\ast}, \nu_{2i}^{\#}).
$
Here, $t_{2i}^\ast = t_{\ell}(m_{2i}^\ast, \Sigma_{2i}^\ast)$ is the true response $t$-distribution whose location and scale are given by $N(m_{2i}^\ast, \Sigma_{2i}^\ast) = \xi_{\nu_{1\oplus}}\langle X_i, \mathbb{B}_0 \rangle_2$ if $C_i=0$, and by $(m_{2i}^\ast, \Sigma_{2i}^\ast) = \langle Z_i, \mathbb{D}_0 \rangle_2$ if $C_i = 1$. $\nu_{2i}^\#$ is the fitted response Gaussian distribution.
We set $d=2, n=200, N=500$ and consider three scenarios with the degree of the freedom $\ell \in \{5, 10, 15\}$. 
As with the previous experiments, we simulate 500 runs, and for each Monte Carlo run, we compute the AWD \eqref{eq:awd} based on 200 new predictors. 
The results of the proposed and alternative methods are summarized in the boxplots of Figure \ref{fig:awd_comparison_miss}.
In all three scenarios, the proposed method outperforms the alternative method. In addition, the prediction performance is getting better as the degree of freedom increases. This result comes  from the fact that as the degree of freedom increases, the $t$-distribution becomes more close to the Gaussian distribution, and thus there is less model misspecification. 

\begin{figure}
    \centering
\includegraphics[width=120mm]{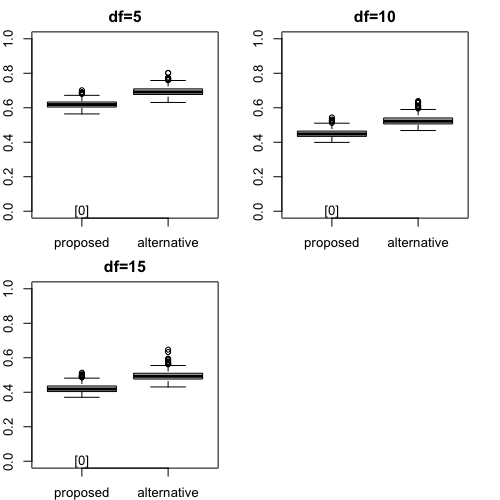}
    \caption{Boxplots of the out-of-sample AWDs defined as \eqref{eq:awd} for the three scenarios with the degree of the freedom $\ell \in \{5, 10, 15\}$. "proposed" denotes the proposed method and "alternative" denotes the alternative method. The number in brackets "[ ]" below the boxplots for the proposed indicates how many runs event \eqref{eq:fall} happened and boundary projection was needed.}
    \label{fig:awd_comparison_miss}
\end{figure}

\section{Applications} \label{sec:application}
In this section, we employ the proposed regression model to grasp the relationship between daily weather in spring (March, April, and May) and that in summer (Jun, July, and August) in Calgary, Alberta.
We obtain data from \url{https://calgary.weatherstats.ca}. This dataset contains the temperature and humidity for each day in Calgary from $1953$ to $2021$. We consider the joint distribution of the average temperatures recorded daily and the average relative humidity recorded daily. 
We regard each pair of daily values as one observation from a two-dimensional Gaussian distribution.
As examples, 
Figure \ref{fig:real_data_example} illustrates the observations and estimated Gaussian  densities for spring and summer in each year from $1953$ to $1956$. 

\begin{figure}
    \centering
\includegraphics[width=140mm]{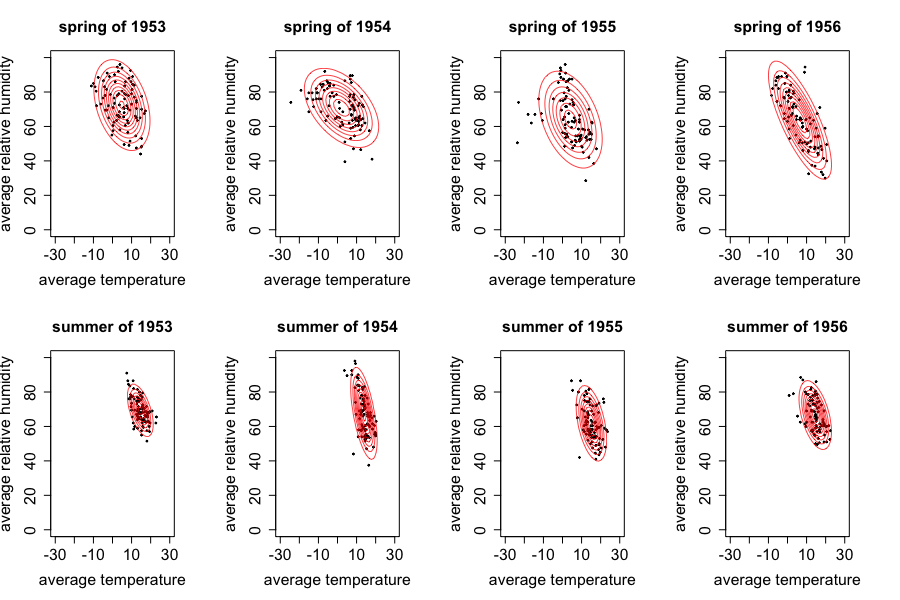}
    \caption{Observed data and estimated Gaussian joint densities of the average temperatures and average relative humidity in spring (top row) and summer (bottom row) from 1953 to 1956. 
    Black points are observed data and solid lines are contour lines of estimated densities.}
    \label{fig:real_data_example}
\end{figure}

We applied the proposed  \eqref{eq:model} and  alternative \eqref{eq:alt_model}  regression models 
with the distributions for spring as the predictor and summer as the response. Models are trained on data up to $1988$ and predictions are computed for the remaining period, where we predicted the distribution of summer based on that of spring for each year. 

Table \ref{tab:five_sum_prop} shows the fitting and prediction results of the proposed method for training and prediction periods. Additionally, Table \ref{tab:five_sum_alt} shows the result of the alternative method. In these tables, we report the summary of the Wasserstein discrepancies between observed and fitted distributions in training periods, and those between observed and predicted distributions in prediction periods. We also show the prediction results of both methods from 2017 to 2019 in Figure \ref{fig:real_data_pred}.
We find that fitting and prediction by the proposed model are generally better than those by the alternative model. This result can be explained by the fact that the proposed model takes into consideration the geometry of the Wasserstein space while the alternative model does not.

\begin{table}[htbp]
    \centering
     \begin{tabular}{lccccc} \hline
    & Min & $Q_{0.25}$ & Median & $Q_{0.75}$ & Max \\ \hline
   Training & 0.5725 & 1.7709 & 3.0337 & 4.5545 & 6.4389  \\
   Prediction & 1.708 & 2.748 & 3.991 & 5.606 & 12.401  \\ \hline
\end{tabular}
    \caption{Summary of the Wasserstein discrepancies for the proposed method in training and prediction periods.}
    \label{tab:five_sum_prop}
\end{table}

\begin{table}[htbp]
    \centering
    \begin{tabular}{lccccc} \hline
    & Min & $Q_{0.25}$ & Median & $Q_{0.75}$ & Max \\ \hline
   Training & 0.3086 & 2.3041 & 3.2879 & 4.7202 & 6.8268    \\
   Prediction & 1.317 & 3.610 & 5.409  & 7.306  & 10.513   \\ \hline
\end{tabular}
    \caption{Summary of the Wasserstein discrepancies for the alternative method in training and prediction periods.}
    \label{tab:five_sum_alt}
\end{table}

\begin{figure}[htbp]
    \centering
\includegraphics[width=140mm]{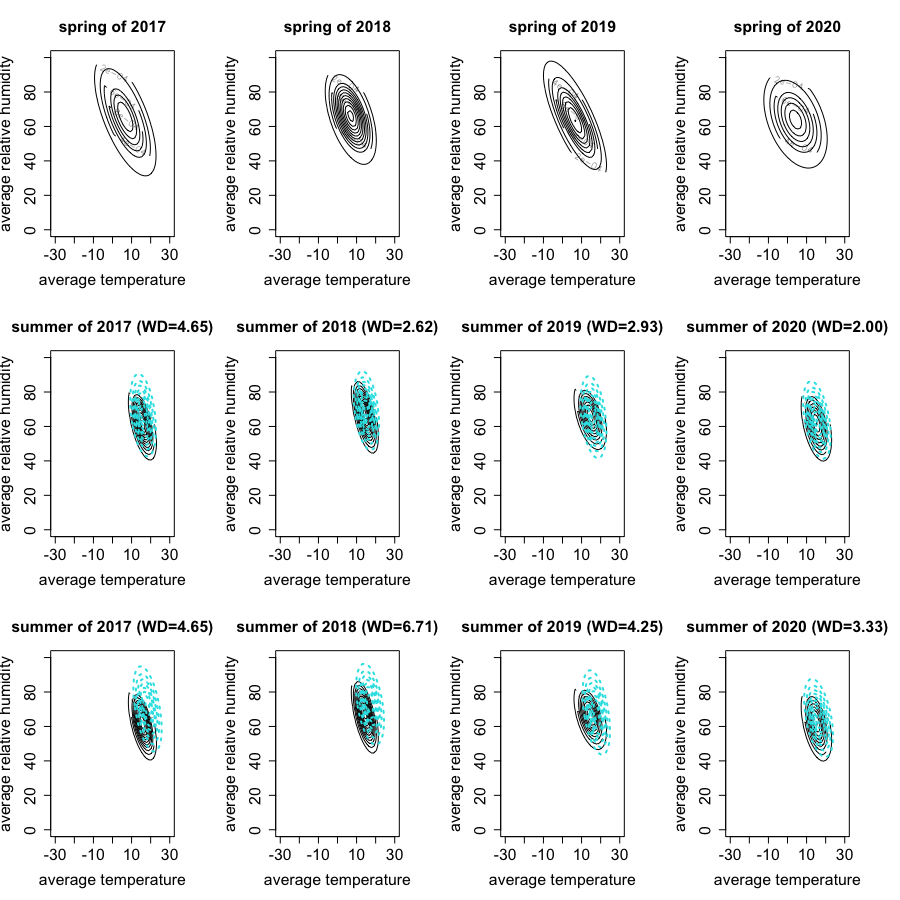}
    \caption{Observed and predicted (middle and bottom rows) densities of the average temperatures and average relative humidity in spring (top row) and summer (middle and bottom rows) from 2017 to 2020. Solid lines are contour lines of observed densities, and dashed lines (middle and bottom rows) are contour lines of predicted densities. 
    Predictions in the middle row are by the proposed method, while those in the bottom row are by the alternative method.
    In the middle and bottom rows, the Wasserstein discrepancies (WDs) between observed and predicted densities are also listed.}    \label{fig:real_data_pred}
\end{figure}

\section{Conclusion} \label{sec:conclusion} 
In this paper, we propose the distribution-on-distribution regression models for multivariate Gaussians with the Wasserstein metric. 
In the proposed regression models, Gaussian distributions are transformed into elements in linear matrix spaces by the proposed nearly isometric maps, and the regression problem comes down to matrix-on-matrix linear regression. 
It has the advantage that the distribution-on-distribution regression is reduced to a linear regression while keeping the properties of distributions.
Also, owing to the linear regression structure, we can easily implement and interpret the models.
We incorporate a low-rank structure in the parameter tensor to address large dimensional Gaussian distributions and also discuss the generalization of our models to the class of elliptically symmetric distributions.
In the simulation studies, we find that our models perform better than an alternative approach of transforming Gaussian distributions to matrices that do not consider the Wasserstein metric.

\appendix

\begin{center}
 {\bf \large Appendix}
 \end{center}

\section{Proofs}
\begin{proof}[Proof of  Proposition 1]
    Firstly, we set $a = m - S(\Sigma_\ast, \Sigma)m_\ast$ and 
    $V = S(\Sigma_\ast, \Sigma) - I$. Then, we have
    $
    a + Vm_\ast 
    =
    m - m_\ast
    $ and 
    \begin{gather*}
         V\Sigma_\ast V
    =
    \Sigma + \Sigma_\ast - \Sigma_\ast^{1/2}[\Sigma_\ast^{1/2} \Sigma \Sigma_\ast^{1/2}]^{1/2}\Sigma_\ast^{-1/2}
     - \Sigma_\ast^{-1/2}[\Sigma_\ast^{1/2} \Sigma \Sigma_\ast^{1/2}]^{1/2}\Sigma_\ast^{1/2}.
    \end{gather*}
    Therefore, $ \|\varphi_{\mu_\ast}\mu\|_{(m_\ast, \Sigma_\ast)}^2$ is expressed as
  \begin{align*}
      \|\varphi_{\mu_\ast}\mu\|_{(m_\ast, \Sigma_\ast)}^2
      &=\|a+Vm_\ast\|^2+\text{tr}(V\Sigma_\ast V) \\
    &=
    \|m - m_\ast\|^2
    +
    \text{tr}(\Sigma) + \text{tr}(\Sigma_\ast)
    -
    \text{tr}(\Sigma_\ast^{1/2}[\Sigma_\ast^{1/2} \Sigma \Sigma_\ast^{1/2}]^{1/2}\Sigma_\ast^{-1/2}) \\ 
    &\,\,\,\,\,\,\,\ - \text{tr}(\Sigma_\ast^{-1/2}[\Sigma_\ast^{1/2} \Sigma \Sigma_\ast^{1/2}]^{1/2}\Sigma_\ast^{1/2}) \\
    &=\|m - m_\ast\|^2
    + \text{tr}(\Sigma) + \text{tr}(\Sigma_\ast) 
    -2\text{tr}([\Sigma_\ast^{1/2} \Sigma \Sigma_\ast^{1/2}]^{1/2}) \\ 
    &=d_W^2(\mu, \mu_\ast).
  \end{align*}
  Next, let $U$ be a $d \times d$ orthogonal matrix and suppose 
  $\mu_\ast = N(m_\ast, \Sigma_\ast), \mu_1 = N(m_1, \Sigma_1)$ and $\mu_2 = N(m_2, \Sigma_2)$ are Gaussian measures in $\mathscr{C}_U$. 
  Because $\Sigma_1^{1/2}\Sigma_2^{1/2} = \Sigma_2^{1/2}\Sigma_1^{1/2}$ holds in this setting, the Wasserstein distance between $\mu_1$ and $\mu_2$ is expressed as 
  \begin{equation}
      d_W^2(\mu_1, \mu_2)
      =
     \|m_1 - m_2\|^2 + \text{tr}((\Sigma_1^{1/2} - \Sigma_2^{1/2})^2).
     \label{eq:diag_1}
  \end{equation}
  On the other hand, 
  because $\Sigma_\ast^{1/2} \Sigma_1^{1/2} = \Sigma_1^{1/2}\Sigma_\ast^{1/2}$ and 
  $\Sigma_\ast^{1/2} \Sigma_2^{1/2} = \Sigma_2^{1/2}\Sigma_\ast^{1/2}$ also hold in this setting, we have
  \begin{gather*}
      \varphi_{\mu_\ast}\mu_1 = (m_1 - \Sigma_1^{1/2}\Sigma_\ast^{-1/2}m_\ast, \Sigma_1^{1/2}\Sigma_\ast^{-1/2} - I), \\ 
    \varphi_{\mu_\ast}\mu_2 = (m_2 - \Sigma_2^{1/2}\Sigma_\ast^{-1/2}m_\ast, \Sigma_2^{1/2}\Sigma_\ast^{-1/2} - I).
  \end{gather*}
  This implies 
  \[
  \varphi_{\mu_\ast}\mu_1 - \varphi_{\mu_\ast}\mu_2 
  =
  (m_1 - m_2 - (\Sigma_1^{1/2} - \Sigma_2^{1/2})\Sigma_\ast^{-1/2}m_\ast, (\Sigma_1^{1/2} - \Sigma_2^{1/2})\Sigma_\ast^{-1/2}), 
  \]
  and we have
  \begin{equation}
      \| \varphi_{\mu_\ast}\mu_1 - \varphi_{\mu_\ast}\mu_2 \|_{(m_\ast, \Sigma_\ast)}^2
  =
  \|m_1 - m_2\|^2 + \text{tr}((\Sigma_1^{1/2} - \Sigma_2^{1/2})^2).
  \label{eq:diag_2}
  \end{equation}
 From \eqref{eq:diag_1} and \eqref{eq:diag_2}, we obtain
 $d_W(\mu_1, \mu_2)
=
\|\varphi_{\mu_\ast}\mu_1 - \varphi_{\mu_\ast}\mu_2 \|_{(m_\ast, \Sigma_\ast)}$.
\end{proof}

To prove Theorem 1 and 2, we employ the following general result regarding the in-sample prediction error of least squares regression, which is shown by \cite{park2023towards}. 
We refer to Section A.2 in \cite{park2023towards} for Gaussian random variables in Hilbert spaces.

\begin{thm}[\cite{park2023towards}, Section 4.1]
Let  $x_1, ..., x_n$ be fixed covariates taking values in a set $\mathcal{X}$, and let $Y_1, ..., Y_n$ be random variables taking values in a separable Hilbert space $(\mathcal{Y}, \|\cdot\|_\mathcal{Y})$ satisfying 
$
Y_i = g_0(x_i) + \varepsilon_i, i=1, ..., n.
$
Here, $\varepsilon_i$ are independent Gaussian noise terms with zero mean  and covariance trace $1$, and $g_0: \mathcal{X} \to \mathcal{Y}$ is an unknown function in a class $\mathcal{G}$. Let define the empirical norm
$\|g\|_n = \sqrt{n^{-1}\sum_{i=1}^n \|g(x_i)\|_\mathcal{Y}^2}$ for $g \in \mathcal{G}$, and define
$J(\delta) = \int_0^\delta \sqrt{\log N_n(t, \mathcal{B}_n(\delta; \mathcal{G}), \|\cdot\|_n)}dt$ for $\delta > 0$, where $N_n(t, \mathcal{B}_n(\delta; \mathcal{G}), \|\cdot\|_n)$ is the $t$-covering number of the ball
$\mathcal{B}_n(\delta; \mathcal{G}) = \{g \in \mathcal{G}: \|g\|_n \le \delta\}$. Then, if there exist real sequence $\{\delta_n\}$ and constant $C > 0$ such that $J(\delta_n) \le C \sqrt{n}\delta_n^2$, the least squares estimator $\hat{g}_n = \argmin_{g \in \mathcal{G}}n^{-1}\sum_{i=1}^n \|Y_i - g(x_i) \|_\mathcal{Y}^2$ satisfies 
$\|\hat{g}_n - g_0\|_n = O_P(\delta_n)$.
 \label{thm:pred_error_hilbert}
\end{thm}

Using this result, we prove Theorem 1 and 2. 
Throughout the proofs, we denote $a \lesssim b$
when there exists a constant $C > 0$ not depending on $n, d_1, d_2, K$ such that $a \le Cb$.
\begin{proof}[Proof of Theorem 1]
Firstly we bound the in-sample prediction error regarding the map $\Gamma_{\mathbb{B}_0}$, which is defined by 
\begin{align}
    \|\Gamma_{\tilde{\mathbb{B}}} - \Gamma_{\mathbb{B}_0}\|_n
    =
    \sqrt{
    n^{-1}\sum_{i=1}^n \|\Gamma_{\tilde{\mathbb{B}}}(X_i) - \Gamma_{\mathbb{B}_0}(X_i) \|_{({m}_{2\oplus}, {\Sigma}_{2\oplus})}^2}.
\end{align}
    Our strategy is to bound the metric entropy of the function space ${\mathscr{F}} = \{\Gamma_{\mathbb{B}}: \mathbb{B} \in {\mathcal{B}} \}$ and employ Theorem \ref{thm:pred_error_hilbert}.
    We define the $\delta$-ball of space $\mathscr{F}$  as $\mathcal{B}_n(\delta; {\mathscr{F}}) = \{ \Gamma_{\mathbb{B}} \in {\mathscr{F}}: \|\Gamma_{\mathbb{B}}\|_n \le \delta\}$ and denote its $t$-covering number as 
    $N_n(t, \mathcal{B}_n(\delta; {\mathscr{F}}), \|\cdot\|_n)$.
    By defining 
       $\|\mathbb{B}\|' = \|\Gamma_{\mathbb{B}}\|_n$ for $\mathbb{B} \in {\mathcal{B}}$,
    the set $\mathcal{B}_n(\delta; \mathscr{{F}})$ is isometric to the $\delta$-ball within the space $({\mathcal{B}}, \|\cdot\|')$.
    Since the space $({\mathcal{B}}, \|\cdot\|')$  has dimension $d_1(d_1+1)d_2(d_2+3)/2$, by a volume ratio argument (Example 5.8 in \cite{wainwright2019high}), we have
    \[
    \log N_n(t, \mathcal{B}_n(\delta; {\mathscr{F}}), \|\cdot\|_n)
    \lesssim d_1^2d_2^2 \log \left(1 + \frac{2\delta}{t}\right). 
    \]
    Using this upper bound, we have
    \begin{align*}
    \int_0^\delta \sqrt{\log N_n(t, \mathcal{B}_n(\delta; {\mathscr{F}}), \|\cdot\|_n)}dt 
    &\lesssim d_1d_2
    \int_0^\delta \sqrt{ \log\left(1 + \frac{2\delta}{t}\right)}dt  \\
    &= \delta d_1d_2 \int_0^1 \sqrt{\log\left(1 + \frac{2}{u}\right)}du \quad (u=t/\delta)  \\
    &\lesssim \delta d_1d_2.
    \end{align*}
This implies we can apply Theorem \ref{thm:pred_error_hilbert} with $\delta_n = d_1d_2/\sqrt{n}$ and obtain $\|\Gamma_{\tilde{\mathbb{B}}} - \Gamma_{\mathbb{B}_0}\|_n = O_P(d_1d_2/\sqrt{n})$.

Next, we bound the in-sample prediction error $\mathcal{R}_n({\Gamma}_{\mathcal{G}, \tilde{\mathbb{B}}}, {\Gamma}_{\mathcal{G}, {\mathbb{B}}_0})$. 
Because the Wasserstein space has nonnegative sectional curvature at any reference measure (e.g., Section 2.3.2 in \cite{panaretos2020invitation}), the Gaussian space, which is the restriction of the Wasserstein space to Gaussian measures, also has this property. 
In other words, the inequality
\[
d_W(\mu_1, \mu_2) \le \|\varphi_{\nu_{2\oplus}}\mu_1 - \varphi_{\nu_{2\oplus}}\mu_2\|_{(m_{2\oplus}, \Sigma_{2\oplus})}
\]
holds for any $\mu_1, \mu_2 \in \mathcal{G}(\mathbb{R}^{d_2})$.
This implies  $\mathcal{R}_n({\Gamma}_{\mathcal{G}, \tilde{\mathbb{B}}}, {\Gamma}_{\mathcal{G}, {\mathbb{B}}_0}) \le \|\Gamma_{\tilde{\mathbb{B}}} - \Gamma_{\mathbb{B}_0}\|_n$ holds, and combining this fact with $\|\Gamma_{\tilde{\mathbb{B}}} - \Gamma_{\mathbb{B}_0}\|_n = O_P(d_1d_2/\sqrt{n})$, we have $\mathcal{R}_n({\Gamma}_{\mathcal{G}, \tilde{\mathbb{B}}}, {\Gamma}_{\mathcal{G}, {\mathbb{B}}_0}) = O_P(d_1d_2/\sqrt{n})$.
\end{proof}

\begin{proof}[Proof of  Theorem 2]

As with the proof of Theorem 1, we firstly bound the in-sample prediction error regarding the map $\Gamma_{\mathbb{B}_0}$. We define the function space as ${\mathscr{F}}_{\text{low}} = \{\Gamma_{\mathbb{B}}: \mathbb{B} \in {\mathcal{B}}_{\text{low}} \}$, define its $\delta$- ball as $\mathcal{B}_n(\delta; {\mathscr{F}}_{\text{low}}) = \{ \Gamma_{\mathbb{B}} \in {\mathscr{F}}_{\text{low}}: \|\Gamma_{\mathbb{B}}\|_n \le \delta\}$ , and denote its $t$-covering number as 
$N_n(t, \mathcal{B}_n(\delta; {\mathscr{F}}_{\text{low}}), \|\cdot\|_n)$.
By defining 
       $\|\mathbb{B}\|'' = \|\Gamma_{\mathbb{B}}\|_n$ for $\mathbb{B} \in {\mathcal{B}}_{\text{low}}$,
    the set $\mathcal{B}_n(\delta; \mathscr{{F}}_{\text{low}})$ is isometric to the $\delta$-ball within the space $({\mathcal{B}}_{\text{low}}, \|\cdot\|'')$.
Recall that if a tensor $\mathbb{B}$ is in $\mathcal{B}_{\text{low}}$,
there exist matrices $A_1 \in \mathbb{R}^{d_1\times K}, A_2 \in \mathbb{R}^{(d+1)\times K}, A_3 \in \mathbb{R}^{d_2 \times K}, A_4 \in \mathbb{R}^{(d_2+1)\times K}$ such that $\mathbb{B} = (\mathbb{A} + \mathbb{A}^\ast)/2$ with $\mathbb{A} = \llbracket A_1, A_2, A_3, A_4 \rrbracket$.
Let consider an corresponding from $\mathbb{R}^{2K(d_1+d_2+1)}$ to $\mathcal{B}_{\text{low}}$ such that 
\[
(\text{vec}(A_1), \text{vec}(A_2), \text{vec}(A_3), \text{vec}(A_4)) \mapsto 
(\mathbb{A} + \mathbb{A}^\ast)/2, 
\]
where $\mathbb{A} = \llbracket A_1, A_2, A_3, A_4 \rrbracket$.
Moreover, let define
\[
\|(\text{vec}(A_1), \text{vec}(A_2), \text{vec}(A_3), \text{vec}(A_4))\|'''
=
\|(\mathbb{A} + \mathbb{A}^\ast)/2\|''.
\]
Since the $\delta$-ball within the space 
$(\mathcal{B}_{\text{low}}, \|\cdot\|'')$ is isometric to the $\delta$-ball within
$(\mathbb{R}^{2K(d_1+d_2+1)}, \|\cdot\|''')$, we eventually have that the set $\mathcal{B}_n(\delta; \mathscr{{F}}_{\text{low}})$ is isometric to the $\delta$-ball within the space $(\mathbb{R}^{2K(d_1+d_2+1)}, \|\cdot\|''')$. 
Therefore, by a volume ratio argument, we have
    \[
    \log N_n(t, \mathcal{B}_n(\delta; {\mathscr{F}}_{\text{low}}), \|\cdot\|_n)
    \lesssim K(d_1+d_2) \log \left(1 + \frac{2\delta}{t}\right). 
    \]
Using this upper bound, as with the proof of Theorem 1, we have 
\[
\int_0^\delta \sqrt{\log N_n(t, \mathcal{B}_n(\delta; {\mathscr{F}}), \|\cdot\|_n)}dt
\lesssim \delta \sqrt{K(d_1+d_2)}.
\]
This implies
we can apply Theorem \ref{thm:pred_error_hilbert} with $\delta_n = \sqrt{K(d_1+d_2)}/\sqrt{n}$ and obtain $\|\Gamma_{\tilde{\mathbb{B}}} - \Gamma_{\mathbb{B}_0}\|_n = O_P(\sqrt{K(d_1+d_2)}/\sqrt{n})$.

As with the proof of Theorem 1, the nonnegativity of sectional curvature of the Wasserstien space implies $\mathcal{R}_n({\Gamma}_{\mathcal{G}, \tilde{\mathbb{B}}}, {\Gamma}_{\mathcal{G}, {\mathbb{B}}_0}) \le \|\Gamma_{\tilde{\mathbb{B}}} - \Gamma_{\mathbb{B}_0}\|_n$.
Combing this fact with $\|\Gamma_{\tilde{\mathbb{B}}} - \Gamma_{\mathbb{B}_0}\|_n = O_P(\sqrt{K(d_1+d_2)}/\sqrt{n})$, we obtain $\mathcal{R}_n({\Gamma}_{\mathcal{G}, \tilde{\mathbb{B}}}, {\Gamma}_{\mathcal{G}, {\mathbb{B}}_0}) = O_P(\sqrt{K(d_1+d_2)}/\sqrt{n})$.

\end{proof}

\section{Parameter Identification}
In this section, we deal with the identification of regression parameter $\mathbb{B}$ in our proposed models. Although the parameter $\mathbb{B}$ does not need to be identified in the empirical risk minimization problems in the main article, it must be identified when we consider estimation or inference for the regression parameter.

\subsection{Basic Model}
Recall that assuming linear regression model \eqref{eq:model} is equivalent to assuming the model \eqref{eq:model_2} for each $1 \le r \le d_2$ and $1 \le s \le d_2+1$. Let fix indexes $1 \le r \le d_2$ and $1 \le s \le d_2+1$ and consider the identification of parameter $\mathbb{B}[\cdot, \cdot, r, s] \in \mathbb{R}^{d_1 \times (d_1+1)}$
in \eqref{eq:model_2}.
In order to deal with the identifiability issue coming from the symmetry in the matrix $X \in \Xi_{d_1}$, we impose the following condition on the parameter $\mathbb{B}[\cdot, \cdot, r, s]$: 
\begin{align}
    \mathbb{B}[p, q, r, s]
    =
    0, \quad \text{for} \,\ 1 \le p \le d_1, p+2 \le q \le d_2+1.
    \label{eq:ident_cond_basic}
\end{align}
In other words, the matrix $ \mathbb{B}[\cdot, \cdot, r, s]$
has a lower triangular form 
\begin{align}
    \begin{pmatrix} 
  \ast   & \vdots  & \ast &            &  {\Huge O}      & \\ 
  \vdots & \vdots  &      &     \ddots &        & \\
  \vdots & \vdots  &      &            & \ddots & \\
  \ast   & \vdots  &      &    {\Huge {\text{$\ast$}}}          &        & \ast
  \label{eq:mat_form_tri}
\end{pmatrix}, 
\end{align}
where $\ast$ is some real number.
If two matrices $\mathbb{B}[\cdot, \cdot, r, s]$ and 
$\mathbb{B}'[\cdot, \cdot, r, s]$ satisfy the condition \eqref{eq:ident_cond_basic}, we have
\[
\langle X, \mathbb{B}[\cdot, \cdot, r, s]\rangle = \langle X, \mathbb{B}'[\cdot, \cdot, r, s]\rangle
\,\ \text{for any $X \in \Xi_{d_1}$} \implies 
\mathbb{B}[\cdot, \cdot, r, s]
=
\mathbb{B}'[\cdot, \cdot, r, s],
\]
which guarantees the identifiability of the parameter 
 $\mathbb{B}[\cdot, \cdot, r, s]$.
 
In summary, by adding condition \eqref{eq:ident_cond_basic} to the existing  parameter space,  we define the following modified parameter space for the basic model :
\begin{align}
    \mathcal{B}^\ast = \{\mathbb{B} \in \mathcal{B}: \text{the condition \eqref{eq:ident_cond_basic} holds for each $1 \le r \le d_2$ and $1 \le s \le d_2+1$}\}.
    \label{eq:modif_basic}
\end{align}
Then, the parameter $\mathbb{B}$ is uniquely identified in $\mathcal{B}^\ast$.

\subsection{Low-Rank Model}
Next, we consider the identification of regression parameters in the low-rank model.
Let $\mathbb{B}$ has the form $\mathbb{B} = (\mathbb{A} + \mathbb{A}^\ast)/2$ and $\mathbb{A}$ admit the rank-$K$ decomposition (\ref{eq:decomposition}). Note that $\mathbb{A}^\ast$ is expressed as
\[
\mathbb{A}^\ast
=
\sum_{k=1}^K a_1^{(k)} \circ a_2^{(k)} \circ a_3 ^{(k)} \circ b^{(k)} 
+
\sum_{k=1}^K a_1^{(k)} \circ a_2^{(k)} \circ c^{(k)} \circ d^{(k)}, 
\]
where 
\[
b^{(k)}
=
\begin{pmatrix}a_4^{(k)}[1] \\ 0 \\ \vdots \\ 0 \end{pmatrix}, 
c^{(k)}
=
\begin{pmatrix}a_4^{(k)}[2] \\ a_4^{(k)}[3] \\ \vdots \\ a_4^{(k)}[d_2+1] \end{pmatrix}, 
d^{(k)}
=
\begin{pmatrix}0 \\ a_3^{(k)}[1] \\ \vdots \\ a_3^{(k)}[d_2] \end{pmatrix}.
\]
Therefore, we have 
\begin{align*}
    \mathbb{B}
    &=
    \sum_{k=1}^K a_1^{(k)} \circ a_2^{(k)} \circ a_3 ^{(k)} \circ a_4^{(k)}/2
    +
    \sum_{k=1}^K a_1^{(k)} \circ a_2^{(k)} \circ a_3 ^{(k)} \circ b^{(k)}/2
    +
    \sum_{k=1}^K a_1^{(k)} \circ a_2^{(k)} \circ c ^{(k)} \circ d^{(k)}/2 \\
    &=
    \sum_{k=1}^K a_1^{(k)} \circ a_2^{(k)} \circ a_3 ^{(k)} \circ (a_4^{(k)}+b^{(k)})/2 
    +
    \sum_{k=1}^K a_1^{(k)} \circ a_2^{(k)} \circ c ^{(k)} \circ d^{(k)}/2, 
\end{align*}
which means $\mathbb{B}$ admits the rank-$2K$ decomposition. Let define matrices $B_1 \in \mathbb{R}^{d_1 \times (2K)}, B_2 \in \mathbb{R}^{(d_1+1) \times (2K)}, B_3 \in \mathbb{R}^{d_2 \times (2K)}, B_4 \in \mathbb{R}^{(d_2+1)^\times (2K)}$ as
$
    B_1 = [a_1^{(1)}, ..., a_1^{(K)}, a_1^{(1)}, ..., a_1^{(K)}], 
    B_2 = [a_2^{(1)}, ..., a_2^{(K)}, a_2^{(1)}, ..., a_2^{(K)}], 
    B_3 = [a_3^{(1)}, ..., a_3^{(K)}, c^{(1)}, ..., c^{(K)}], B_4 = [(a_4^{(1)}+b^{(1)})/2, ..., (a_4^{(K)}+b^{(K)})/2, d^{(1)}/2, ..., d^{(K)}/2]
$. 
Then, we have $\mathbb{B} = \llbracket B_1, B_2, B_3, B_4 \rrbracket$. 
Following an identification strategy used in \cite{zhou2013tensor} for tensor regression models,
 we adopt
 the following specific constrained parametrization 
 to fix the scaling and permutation indeterminacy of the tensor decomposition.
\begin{itemize}
    \item To fix the scaling indeterminacy, we assume 
     \begin{align}
         a_1^{(k)}[1] = a_2^{(k)}[1] = a_3^{(k)}[1] = a_4^{(k)}[1] = 1, \quad 1 \le k \le K. \label{eq:scale_cond}
     \end{align}
     In other words, the first rows of  $B_1, B_2, B_3$ are ones. This scaling of $B_1, B_2, B_3$ determines the first row of $B_4$ and fixes scaling indeterminacy (Section 4.2 in \cite{zhou2013tensor}).
     \item To fix the permutation indeterminacy, we assume that the last row elements of $B_4$ are distinct and arranged in the descending order 
     \begin{align}
         a_4^{(1)}[d_2+1] > \cdots > a_4^{(K)}[d_2+1] > a_3^{(1)}[d_2] > \cdots a_3^{(K)}[d_2].
          \label{eq:permutation_cond}
     \end{align}
      This fixes permutation indeterminacy (Section 4.2 in \cite{zhou2013tensor}).
\end{itemize}
 Adding these constraints to the existing parameter space, 
 we define the modified parameter space for the rank-$K$ model as
 \begin{align}
     \mathcal{B}_{\text{low}}^\ast
     =
     \{\mathbb{B} = (\mathbb{A} + &\mathbb{A}^\ast)/2 \in \mathcal{B}_{\text{low}}: \mathbb{A} = \llbracket A_1, A_2, A_3, A_4 \rrbracket,  \notag \\
     & A_1, A_2, A_3, A_4 \,\ \text{satisfy the conditions}\, \eqref{eq:scale_cond}, \eqref{eq:permutation_cond}\}.
     \label{eq:modif:low}
\end{align}
If the tensor $\mathbb{B} = (\mathbb{A} + \mathbb{A}^\ast)/2 \in \mathcal{B}_{\text{low}}^\ast$ satisfies the condition 
\begin{align}
     \text{rank} B_1 + 
     \text{rank} B_2 + 
     \text{rank} B_3 + 
     \text{rank} B_4 \ge 4K+3, 
     \label{rank_cond}
 \end{align}
 then Proposition 3 in \cite{zhou2013tensor} implies that  $\mathbb{B}$ is uniquely identified in $\mathcal{B}_{\text{low}}^\ast$.

\section{Consistency and Asymptotic Normality of Estimators}
In this section, we study the asymptotic property of estimators for the regression parameter in the basic model.
Let $\{(\nu_{i1}, \nu_{i2})\}_{i=1}^n$ be independent realization of the pair of Gaussian distributions $(\nu_1, \nu_2)$ from the basic model. For simplicity, we assume the true Fr\'echet means $\nu_{1\oplus}, \nu_{2\oplus}$ are known and distributions $\{(\nu_{1i}, \nu_{2i})\}_{i=1}^n$ are fully observed. 

We set $X_i = \varphi_{\nu_{1\oplus}}\nu_{1i}, Y_i = \varphi_{\nu_{2\oplus}}\nu_{2i}$ and define an estimator as 
$
\tilde{\mathbb{B}}_n = 
\argmin_{\mathbb{B} \in \mathcal{B}^\ast} \sum_{i=1}^n \|{Y}_i - \langle {X}_i, \mathbb{B} \rangle \|_{(m_{2\oplus}, \Sigma_{2\oplus})}^2 
$. 
Here, $\mathcal{B}^\ast$ is the modified parameter space defined by \eqref{eq:modif_basic}.

In order to state our results, we introduce a half-vectorization of tensor $\mathbb{B}$ in $\mathcal{B}^\ast$. For a matrix $A \in \mathbb{R}^{d \times (d+1)}$ , we define its vectorization $\mathrm{vech}^\ast(A) \in \mathbb{R}^{d(d+3)/2}$ as 
\begin{align*}
    \mathrm{vech}^\ast(A) 
=
&(A[1, 1], A[2, 1], \cdots ,A[d, 1], A[1, 2], A[2, 2], \cdots A[d, 2], \\
&A[2, 3], \cdots ,A[d, 3], A[3, 4], \cdots ,A[d,4], \cdots ,A[d-1, d], A[d, d], A[d, d+1]).
\end{align*}
Furthermore, for a tensor $\mathbb{B} \in \mathcal{B}^\ast$, we define its vectorization $\text{vec}^\ast(\mathbb{B}) \in \mathbb{R}^{d_1(d_1+1)d_2(d_2+1)/4}$ as 
\[
\text{vec}^\ast(\mathbb{B})
=
(  (\mathrm{vech}^\ast(\mathbb{B}[\cdot, \cdot, r, s])^\top)_{1 \le r \le d_2,r+2 \le s \le d_2+1})^\top.
\]
Note that the $\text{vec}^\ast(\cdot)$ 
operator is a one-to-one correspondence between $\mathcal{B}^\ast$ and $\mathbb{R}^{d_1(d_1+1)d_2(d_2+1)/4}$. 
Therefore, for any $\theta \in \mathbb{R}^{d_1(d_1+1)d_2(d_2+1)/4}$, 
there uniquely exists a tensor $\mathbb{B} \in \mathcal{B}^\ast$ such that $\text{vec}^\ast(\mathbb{B}) = \theta$.
We denote this tensor $\mathbb{B}$ as $\mathbb{B}(\theta)$.

Under this vectorization, we denote 
$\tilde{\theta}_n = \text{vec}^\ast(\tilde{\mathbb{B}}_n)$ and  $\theta_0 = \text{vec}^\ast({\mathbb{B}}_0)$, and analyze the asymptotic property of the estimator $\tilde{\theta}_n$ with the standard theory for M-estimation. 
For vector $\theta \in \mathbb{R}^{d_1(d_1+1)d_2(d_2+1)/4}$ and 
matrices $X \in \Xi_{d_1}, Y \in \Xi_{d_2}$, 
we define 
\[
m_{\theta}(\mathrm{vech}^\ast(X), \mathrm{vech}^\ast(Y)) = \|\mathrm{vech}^\ast(Y) - \mathrm{vech}^\ast(\langle X, \mathbb{B}(\theta) \rangle_2)\|_{(m_{2\oplus}, \Sigma_{2\oplus})}^2.
\]
Here, for a vector $z \in \mathbb{R}^{d_2(d_2+3)/2}$ represented as $z = \mathrm{vech}^\ast(A)$ with a matrix $A \in \mathbb{R}^{d_2(d_2+3)}$, we 
define its norm as $\|z\|_{m_{2\oplus}, \Sigma_{2\oplus}} = \|A\|_{m_{2\oplus}, \Sigma_{2\oplus}}$.
Then, the estimator $\tilde{\theta}_n$ is characterized as the minimizer of the criterion function $\theta \mapsto n^{-1}\sum_{i=1}^n m_{\theta}(\mathrm{vech}^\ast(X_i), \mathrm{vech}^\ast(Y_i))$. Note that 
the vector
$\mathrm{vech}^\ast(\langle X, \mathbb{B}(\theta) \rangle_2) \in \mathbb{R}^{d_2(d_2+3)/2}$ has the form
\[
\mathrm{vech}^\ast(\langle X, \mathbb{B}(\theta) \rangle_2)
=
(\langle \mathrm{vech}^\ast(X), \mathrm{vech}^\ast(\mathbb{B}(\theta)[\cdot, \cdot, r, s])\rangle)_{1 \le r \le d_2, r+2 \le s \le d_2+1},
\]
which implies $\tilde{\theta}_n$ is the least-square estimator in the linear regression model between vectors
$\mathrm{vech}^\ast(X)$ and  $\mathrm{vech}^\ast(Y)$.

Then, we obtain the following results. We denote the partial derivative of the function  $m_\theta$  in terms of $\theta$ as $\nabla_\theta m_\theta$.

\begin{thm}[Consistency of Estimator]
Assume $\theta_0$ is in a compact parameter space  $\Theta_0 \subset \mathbb{R}^{d_1(d_1+1)d_2(d_2+1)/4}$ and the pair of vectors $(\mathrm{vech}^\ast(X_i), \mathrm{vech}^\ast(Y_i))$ is supported on a bounded set. Then,
$\tilde{\theta}_n$ is a consistent estimator for $\theta_0$.
\label{thm:consist}
\end{thm}

\begin{proof}
    We show that the set of functions $\{m_\theta: \theta \in \Theta_0\}$ is a Glivenko-Cantelli class (Section 19 in \cite{van2000asymptotic}). 
    If this holds, the consistency of the estimator $\tilde{\theta}_n$ follows from Theorem 5.7 in \cite{van2000asymptotic}.
    Note that for a vector $z = (z_1, ..., z_{d_2(d_2+3)/2}) \in \mathbb{R}^{d_2(d_2+3)/2}$, the norm  $\|z\|_{m_{2\oplus},\Sigma_{2\oplus}}$ has the form
    \begin{equation}
        \|z\|_{m_{2\oplus}, \Sigma_{2\oplus}}^2
    =
    \sum_{1 \le i \le j \le d_2(d_2+3)/2}
    c_{ij}z_iz_j, 
    \label{eq:norm_form}
    \end{equation}
    where $c_{ij}$ are  constants determined by the values of $m_{2\oplus}$ and $\Sigma_{2\oplus}$. 
    This implies 
    that the map $\theta \mapsto m_\theta(\mathrm{vech}^\ast(X), \mathrm{vech}^\ast(Y))$ is continuous for each fixed $\mathrm{vech}^\ast(X)$ and 
    $\mathrm{vech}^\ast(Y)$. Moreover,
    because the parameter $\theta$ and 
    vectors $\mathrm{vech}^\ast(X)$ and 
    $\mathrm{vech}^\ast(Y)$ are in bounded regions,
    the map $m_\theta$ is also uniformly bounded. 
    That is, there exists a constant $C > 0$ such that 
    $m_\theta(\mathrm{vech}^\ast(X), \mathrm{vech}^\ast(Y)) \le C$ for all $\theta \in \Theta_0, \mathrm{vech}^\ast(X), \mathrm{vech}^\ast(Y)$.
    This implies the set of functions 
    $\{m_\theta: \theta \in \Theta_0 \}$ is dominated by the integrable constant function $C$. 
    Combining these facts with the 
    assumption of compactness of $\Theta_0$, 
    Example 19.8 in \cite{van2000asymptotic} implies that $\{m_\theta: \theta \in \Theta_0\}$ is a Glivenko-Cantelli class.
\end{proof}

\begin{thm}[Asymptotic Normality of Estimator]
In addition to the assumptions in Theorem \ref{thm:consist}, suppose $\theta_0$ is an interior point of $\Theta_0$ and the map $\theta \mapsto \mathbb{E}[m_\theta(\mathrm{vech}^\ast(X_i), \mathrm{vech}^\ast(Y_i))]$ has nonsingular Hessian matrix $V_{\theta_0}$ at $\theta_0$.
Then, $\sqrt{n}(\tilde{\theta}_n - \theta_0) $ converges in distribution to a normal distribution with mean zero and covariance matrix 
\[
V_{\theta_0}^{-1} \mathbb{E}[\nabla_\theta m_{\theta_0}(\mathrm{vech}^\ast(X_i), \mathrm{vech}^\ast(Y_i))  \nabla_\theta m_{\theta_0}(\mathrm{vech}^\ast(X_i), \mathrm{vech}^\ast(Y_i))^\top]V_{\theta_0}^{-1}. 
\]
\label{thm:asymnorm}
\end{thm}

\begin{rmk}
    When the norm $\| \cdot \|_{{(m_{2\oplus}, \Sigma_{2\oplus})}}$ is equal to the Frobenius norm, that is, $m_{2\oplus} = 0$ and $\Sigma_{2\oplus} = I$, the  
    second-derivative matrix $V_{\theta_0}$ has the form 
    \[
    V_{\theta_0}
    =
    \begin{pmatrix}
        \mathbb{E}[\mathrm{vech}^\ast(X_i)\mathrm{vech}^\ast(X_i)^\top] & & {\Huge O} \\
        & \ddots & \\ 
        {\Huge O} & & \mathbb{E}[\mathrm{vech}^\ast(X_i)\mathrm{vech}^\ast(X_i)^\top]
    \end{pmatrix}.
    \]
    Therefore, $V_{\theta_0}$ is nonsingular if and only if the matrix $\mathbb{E}[\mathrm{vech}^\ast(X_i)\mathrm{vech}^\ast(X_i)^\top]$ is nonsingular.
\end{rmk}

\begin{proof}
    We check the conditions of Theorem 5.23 in \cite{van2000asymptotic}, which is a standard result for the asymptotic normality of the M-estimator.
    Noting that the norm $\|z\|_{{(m_{2\oplus}, \Sigma_{2\oplus})}}$ has the form 
    \eqref{eq:norm_form} for a vector  $z = (z_1,..., z_{d_2(d_2+3)/2}) \in \mathbb{R}^{d_2(d_2+3)/2}$, the function $\theta \mapsto m_\theta(\mathrm{vech}^\ast(X), \mathrm{vech}^\ast(Y))$ is differentiable on the interior of $\Theta_0$ for each fixed $\mathrm{vech}^\ast(X)$ and 
    $\mathrm{vech}^\ast(Y)$.
    Moreover, because the parameter $\theta$ and 
    vectors $\mathrm{vech}^\ast(X)$ and 
    $\mathrm{vech}^\ast(Y)$ are in bounded regions, the partial derivative $\nabla_\theta m_\theta$ is also bounded. That is, there exists a constant $M > 0$ such that $\|\nabla_\theta m_\theta(\mathrm{vech}^\ast(X), \mathrm{vech}^\ast(Y))\| \le M$ for all $\theta \in \Theta_0,\mathrm{vech}^\ast(X)$ and $\mathrm{vech}^\ast(Y)$.
    Combining this fact with the multi-dimensional mean value theorem, for every $\theta_1$ and $\theta_2$ in a neighborhood of $\theta_0$, we have 
    \[
    |m_{\theta_1}(\mathrm{vech}^\ast(X), \mathrm{vech}^\ast(Y)) - m_{\theta_2}(\mathrm{vech}^\ast(X), \mathrm{vech}^\ast(Y))|
    \le 
    M\|\theta_1 - \theta_2\|.
    \]
    Finally, the map $\theta \mapsto \mathbb{E}[m_\theta(\mathrm{vech}^\ast(X_i), \mathrm{vech}^\ast(Y_i))]$ is assumed to have nonsingular Hessian matrix $V_{\theta_0}$ at $\theta_0$. Then, the conditions of Theorem 5.23 in \cite{van2000asymptotic} are fulfilled, and we have the conclusion from the theorem.
\end{proof}

\bibliographystyle{plain}
\bibliography{ref}

\end{document}